\documentclass[11pt]{article}
\usepackage[letterpaper, margin=1.1in]{geometry}

\title{Two-Sided Weak Submodularity for Matroid Constrained Optimization and Regression} 

\author{Theophile Thiery\thanks{School of Mathematical Sciences, Queen Mary University of London, London, United Kingdom (\href{mailto:t.f.thiery@qmul.ac.uk}{t.f.thiery@qmul.ac.uk},
    \href{mailto:justin.ward@qmul.ac.uk}{justin.ward@qmul.ac.uk}). This work was supported by the Engineering and Physical Sciences Research Council [EP/T006781/1].} \and Justin Ward\footnotemark[1]}
\date{\today}

\usepackage[colorlinks=true, linkcolor=blue, citecolor=blue, backref=page]{hyperref}
\renewcommand*\backref[1]{\ifx#1\relax \else ($\uparrow$ #1) \fi}

\usepackage{amsfonts, bbm, bm}
\usepackage{amsmath}
\usepackage{amsthm}
\usepackage{amssymb}
\usepackage{mathtools}
\usepackage{enumitem}
\usepackage{microtype}
\usepackage[vlined, ruled, algo2e, algosection]{algorithm2e}
\usepackage[displaymath, mathlines, modulo]{lineno}
\usepackage{hyperref}
\usepackage{subcaption}
\usepackage{floatrow}
\usepackage{stmaryrd}
\usepackage{bookmark}
\usepackage{cases}
\usepackage{mathtools, mleftright}
\usepackage{hyperref}
\usepackage{pgfplots}
\usepackage{thm-restate}
\usepackage[classfont=roman]{complexity}

\pgfplotsset{compat=1.14}
\newcommand{\reals}{\mathbb{R}}
\newcommand{\posreals}{\reals_+}

\DeclareMathOperator*{\expect}{\mathbb{E}}
\DeclareMathOperator*{\cov}{\mathrm{Cov}}
\DeclareMathOperator*{\var}{\mathrm{Var}}
\DeclareMathOperator*{\res}{\mathrm{Res}}
\DeclareMathOperator*{\tr}{\mathrm{tr}}
\newcommand{\e}{\varepsilon}
\newcommand{\transp}{^{T}}
\renewcommand{\vec}{\mathbf}

\newcommand{\vb}{\vec{b}}
\newcommand{\ve}{\vec{e}}

\newcommand{\vy}{\vec{y}}
\newcommand{\vu}{\vec{u}}
\newcommand{\vv}{\vec{v}}
\newcommand{\vw}{\vec{w}}
\newcommand{\vx}{\vec{x}}

\newcommand{\vtheta}{\bm{\theta}}
\newcommand{\veps}{\bm{\varepsilon}}

\newcommand{\cA}{\mathcal{A}}
\newcommand{\cB}{\mathcal{B}}
\newcommand{\cD}{\mathcal{D}}
\newcommand{\cE}{\mathcal{E}}
\newcommand{\cI}{\mathcal{I}}
\newcommand{\cM}{\mathcal{M}}
\newcommand{\cO}{\mathcal{O}}
\newcommand{\calS}{\mathcal{S}}
\newcommand{\calL}{\mathcal{L}}
\newcommand{\cT}{\mathcal{T}}

\newcommand{\cX}{\mathcal{X}}
\newcommand{\hcT}{\hat{\cT}}
\newcommand{\hC}{\hat{C}}
\newcommand{\hb}{\hat{\vb}}

\newcommand{\hX}{\hat{X}}

\newcommand{\mc}[3]{m^{(#1)}_{#2,#3}}

\newcommand{\RR}{\mathbb{R}}

\newcommand{\lmin}{\lambda_{\min}}

\newcommand{\game}{\gamma_{\textrm{e}}}
\renewcommand{\epsilon}{\varepsilon}

\newcommand{\lc}{\!\left\{}
\newcommand{\rc}{\right\}}
\newcommand{\ld}{\!\left[}
\newcommand{\rd}{\right]}
\newcommand{\lb}{\!\left(}
\newcommand{\rb}{\right)}

\DeclarePairedDelimiterX{\inp}[2]{\langle}{\rangle}{#1; #2}

\DeclareMathOperator*{\argmax}{arg\,max}
\newcommand{\card}[1]{\left\vert #1\right\vert}

\newcommand{\bb}{\backslash}

\newcommand{\bigT}[1]{\Theta\lb #1 \rb}
\newcommand{\opt}{O}
\newcommand{\fopt}{f\lb \opt \rb}

\newcommand{\esp}[1]{\mathbb{E}\ld #1 \rd}
\newcommand{\prob}[1]{\mathbb{P} \ld #1 \rd}

\newtheorem{theorem}{Theorem}[section]
\newtheorem*{theorem*}{Theorem}
\newtheorem{lemma}[theorem]{Lemma}

\newtheorem{proposition}[theorem]{Proposition}

\theoremstyle{definition}

\theoremstyle{remark}

\newtheorem{example}[theorem]{Example}

\begin{document}
\maketitle
\begin{abstract}
    We study the following problem: Given a variable of interest, we would like to find a best linear predictor for it by choosing a subset of $k$ relevant variables obeying a matroid constraint. This problem is a natural generalization of subset selection problems where it is necessary to spread observations amongst multiple different classes.
    We derive new, strengthened guarantees for this problem by improving the analysis of the residual random greedy algorithm and by developing a novel distorted local-search algorithm.
    To quantify our approximation guarantees, we refine the definition of weak submodularity by  \cite{Das:2011:Submodular} and introduce the notion of an \emph{upper submodularity ratio}, which we connect to the minimum $k$-sparse eigenvalue of the covariance matrix.
    More generally, we look at the problem of maximizing a set function $f$ with lower and upper submodularity ratio $\gamma$ and $\beta$ under a matroid constraint. For this problem, our algorithms have asymptotic approximation guarantee $\frac{1}{2}$ and $1 - e^{-1}$ as the function is closer to being submodular.
    As a second application, we show that the Bayesian A-optimal design objective falls into our framework, leading to new guarantees for this problem as well.
\end{abstract}

\section{Introduction}
In the subset selection problem for linear regression, we are given a collection $\mathcal{X}$ of predictor variables and a target variable $Z$, as well as known covariances between each pair of variables. The goal is to find a small collection $\mathcal{S} \subseteq{\mathcal{X}}$ of at most $k$ predictor variables that gives the best linear predictor for $Z$. When $|\mathcal{X}|$ is very large, the \emph{forward regression} algorithm is commonly employed as a heuristic. It constructs $\mathcal{S}$ iteratively, at each step adding a variable that greedily maximizes the squared multiple correlation objective.
To explain the success of this approach in practice,  \cite{Das:2008:Algorithms} connected subset selection problems with submodular optimization. They showed that the squared multiple correlation objective function, also known as the $R^2$ objective, is \emph{submodular} in the absence of suppressor variables. Intuitively, a variable $X \in \mathcal{X}$ is a suppressor if there is some other variable $Y \in \mathcal{X}$ such that observing $X$ \emph{increases} the (conditional) correlation between $Y$ and the target variable $Z$. We give an example of such a situation in Example~\ref{ex:suppressor}. Even in the presence of suppressors, \cite{Das:2011:Submodular,DBLP:journals/jmlr/DasK18} showed that a weaker property they deemed \emph{weak submodularity} can be used to derive meaningful guarantees. The \emph{submodularity ratio} $\gamma$ measures how far a function deviates from submodularity when considering the aggregate effect of adding elements.
By treating the forward regression algorithm as a variant of the standard greedy algorithm, they showed that it has approximation guarantee of $(1 - e^{-\gamma})$, where $\gamma \in [0,1]$ can be lower bounded by the smallest $2k$-sparse eigenvalue $\lmin(C_{\mathcal{X}},2k)$ of the covariance matrix for $\mathcal{X}$. 

Here, we consider a natural generalization of this problem in which we must select a subset $\mathcal{S}$ that is independent in a general matroid constraint. Such constraints naturally capture settings in which some observations are mutually exclusive (for example, sensors that may be placed in different configurations) or in which it is desirable or necessary to spread observations amongst multiple different classes (for example by time or location).
In contrast to cardinality constraints, the best known guarantee for maximizing the $R^2$ objective in a general matroid is a randomized $1/(1+\gamma^{-1})^2$-approximation via the \textsc{ResidualRandomGreedy} algorithm due to~\cite{Chen:2018:Weakly}.
However, as $\gamma$ tends to $1$ (i.e.\ as the function becomes closer to a submodular function) this bound tends to only $1/4$, while both \textsc{ResidualRandomGreedy} and the standard greedy algorithm are known to provide a $1/2$ approximation for submodular objectives under a matroid constraint. The state of the art in this setting is a $(1-1/e)$-approximation~\cite{Calinescu:2011ju,Filmus:2014}, which is known to be tight~\cite{Nemhauser:1978dm,Feige:1998gx}.

A key difficulty is that the definition of weak submodularity considers only the effect of \emph{adding} elements to the current solution. In contrast, the analysis of \textsc{ResidualRandomGreedy} as well other state-of-the art procedures for submodular optimization in a matroid requires bounding the losses when elements are \emph{removed} or \emph{swapped} from some solution. Stronger ``element-wise'' notions of weak submodularity have been proposed that allow the adaptation of such algorithms but in general, these notions may give weaker bounds than those obtained when the submodularity ratio $\gamma$ can be utilized instead.

 \subsection{Our results}
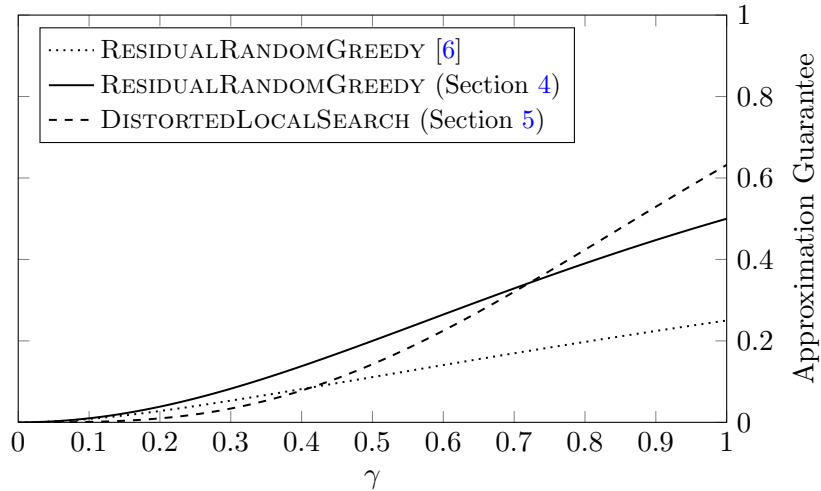
\begin{figure}[t]
\centering
\pgfplotsset{every axis plot/.append style={line width=0.75pt}}
\begin{tikzpicture}
\begin{axis}[
  xlabel={$\gamma$},
  ylabel={Approximation Guarantee},
  xmin=0,
  xmax=1,
  ymin = 0,
  ymax = 1,
  legend cell align=left,
  legend pos=north west,
  ylabel  near ticks,
  yticklabel pos = right,
  width = 11cm,
  height = 7cm
]
  \addplot[domain=0:1,samples=201,dotted] {x^2/(1+x)^2};
  \addplot[domain=0:1,samples=201] {x/(x+1/x)};
  \addplot[domain=0.05:1,samples=201,dashed] {(x^3)*(1-exp(1-1/x-x^2))/(1-x+x^3)};
\legend{\small \textsc{ResidualRandomGreedy}~\cite{Chen:2018:Weakly},\small\textsc{ResidualRandomGreedy} (Section~\ref{sec:impr-analys-rrg}), \small \textsc{DistortedLocalSearch} (Section~\ref{sec:non-oblivious-local})}
\end{axis}
\end{tikzpicture}
\caption{Guarantees for $(\gamma,1/\gamma)$-weakly submodular maximization under a matroid constraint.}
\label{fig:results}
\end{figure}
Here, we give a natural extension of the submodularity ratio $\gamma$ by considering an \emph{upper submodularity ratio} $\beta > 0 $, that bounds how far a function deviates from submodularity when considering the effect of \emph{removing} elements. Intuitively, the parameter $\beta$ compares the loss by removing an entire set compared to the aggregate individual losses for each element.
We show that, as with the submodularity ratio $\gamma$, our $\beta$ can be bounded by spectral quantities in the setting of regression. Specifically, we show that the $R^2$ objective satisfies $\beta \leq \lmin^{-1}(C_\cX,k) \leq \gamma^{-1}$. Since $\beta$ is defined in terms of removing elements from the solution, this requires a different spectral argument than that used to bound $\gamma$ in~\cite{DBLP:journals/jmlr/DasK18}. While their bound for $\gamma$ follows directly by considering an appropriate Rayleigh quotient, here we must relate the average value of the quadratic forms obtained from the inverses of all rank $k-1$ principle submatrices of a matrix $C$ to that obtained from $C^{-1}$.

We also consider the problem of Bayesian A-optimal design, which has been previously studied via weak submodularity by~\cite{Bian:2017:Guarantees,DBLP:conf/icml/HarshawFWK19,Hashemi:2019:Submodular}. In Appendix~\ref{sec:optimal-design}, we show that our parameter $\beta$ can be bounded by $\gamma^{-1}$ for this problem, as well.

Using the connection between subset selection problem and submodular maximization, we consider the more general problem of maximizing a \emph{$(\gamma,\beta)$-weakly submodular} set function for which both ratios are bounded.
We derive improved guarantees for matroid constrained maximization problems via \textsc{ResidualRandomGreedy}. Our guarantee is $\gamma/(\gamma + \beta)$, which approaches $1/2$ as the function $f$ becomes closer to submodular (i.e.\ as $\gamma,\beta \to 1$). It is natural, then, to ask whether it is possible to obtain an algorithm with guarantee approaching the optimal result of $1 - e^{-1}$ for submodular functions. We answer this affirmatively by giving a local search algorithm guided by a distorted potential, and show it achieves a guarantee approaching $(1-e^{-1} - \cO(\epsilon))$ for $(\gamma,\beta)$-weakly submodular functions as $\gamma,\beta \to 1$, where $\epsilon > 0$ is a constant parameter that can be chosen independently of $\gamma,\beta$. Combining these with our spectral bounds, our first analysis improves on the current state of the art for matroid constrained subset selection and A-optimal design problems for all $\gamma > 0$, and our second gives further improvements whenever $\gamma > 0.7217$ (see Figure~\ref{fig:results}).

Our distorted local-search algorithm builds upon similar techniques from the submodular case presented in~\cite{Filmus:2014}. There submodularity of $f$ implies submodularity of the potential $g$, which is used to derive the bounds on $g$ necessary for convergence and sampling, as well the crucial bound linking the local optimality of $g$ to the value of $f$. Here, however, since $f$ is only approximately submodular, these techniques will not work, and so we require a more delicate analysis for each of these components.
A further complication in our setting is that the correct potential $g$ depends on the values of $\gamma$ and $\beta$, which may not be known a priori. We give an approach that is based on guessing the value of a joint parameter in $\gamma$ and $\beta$. Each such guess gives a different distorted potential.  Inspired (broadly) by simulated annealing, we show that if each such new potential is initialized by the local optimum of the previous potential, then the overall running time can be amortized over all guesses.  We present a simplified version of the algorithm and potential function in Section~\ref{sec:non-oblivious-local}, and defer the more technical details to Appendices~\ref{sec:properties-g_phi} and \ref{sec:our-final-algorithm}.

Finally, given the relationship between $\gamma$ and $\beta$ in both problems we consider, it is natural to conjecture that $\beta$ may be bounded in terms of $\gamma$ in some generic fashion for \emph{every} weakly submodular function. However, we show that this is not the case by exhibiting (in Appendix~\ref{sec:how-large-can}) a function on a ground set of size $k$ for which $\beta$ must be $\bigT{k^{1-\gamma}}$.

\subsection{Additional Related Work}

The submodularity ratio $\gamma$ and the corresponding notion of weak submodularity was first introduced to analyze the forward regression and orthogonal matching pursuit algorithms for linear regression by~\cite{DBLP:journals/jmlr/DasK18}. It was latter related to restricted strong convexity by~\cite{Elenberg:2018:Strong}, leading to similar guarantees for generalized linear model (GLM) likelihood, graphical model learning objectives, or an arbitrary M-estimator.  The submodularity ratio has also been applied to the analysis greedy algorithms in other modes of computation~\cite{Khanna:2017:Scalable,Elenberg:2017:StreamWeak}. Together with related algorithmic techniques, it has also lead to algorithms for sensor placement problems \cite{Hashemi:2020:Randomized}, experimental design \cite{DBLP:conf/icml/HarshawFWK19,Bian:2017:Guarantees}, low rank optimization \cite{Khanna:2017:Approximation}, document summarization \cite{Chen:2018:Weakly}, and interpretation of neural networks \cite{Elenberg:2017:StreamWeak}

One motivation for the study of weak submodular functions is to bridge the gap between worst-case theory and the performance of algorithms practice. At the other extreme, one can consider the deviation of a submodular function from linearity. This leads to the notion of \emph{curvature} which can be used to strengthen approximation bounds for both for various submodular optimization problems~\cite{Conforti:1984ig,VondrakCurvature,Sviridenko:2015ur,yoshida_2018,Friedrich_Gobel_Neumann_Quinzan_Rothenberger_2019}, as well as in combination with weak submodularity~\cite{Bian:2017:Guarantees}. Inspired by insights from continuous optimization, \cite{Pokutta:2020:Sharpness} have recently introduced a new notion of \emph{sharpness}, which provides further explanation for the empirically good performance of the greedy algorithm on submodular objectives.


There have been various approaches based on considering \emph{element-wise} bounds on the deviation of a function from submodularity~\cite{DBLP:conf/aistats/BogunovicZC18,Nong:2019:Maximize,Gong:2019:Parametric} including generalizations to functions over the integer lattice~\cite{DBLP:conf/aaai/QianZT018,DBLP:conf/icml/KuhnleSCT18}. These approaches all involve relaxing the notion of decreasing marginal returns by requiring that a function $f$ satisfy $f(A \cup \{e\}) - f(A) \geq \game\cdot (f(B \cup \{e\}) - f(B))$ for all $e \not\in B$ and $A \subseteq B$, where $\game \in [0,1]$ has been variously dubbed the \emph{inverse curvature}~\cite{DBLP:conf/aistats/BogunovicZC18}, \emph{DR ratio}~\cite{DBLP:conf/icml/KuhnleSCT18}, or \emph{generic submodularity ratio}~\cite{Nong:2019:Maximize}. For such functions, it is easy to show that our parameters satisfy $\gamma \geq \game$ and $\beta \leq 1/\game$. Unfortunately, as we show in Section 3, the resulting inequalities may be very far from tight in our setting. In particular, analyses relying on $\game$ may fail to give any non-trivial approximation bounds for regression problems, even in situations when $\lmin(C_\cX)$ and the submodularity ratio $\gamma$ are positive. This observation has motivated our consideration of the more general parameter $\beta$, which allows spectral bounds to be utilized.

Finally, we note that all the definitions introduced here assume that the objective is monotone. Recently, \cite{Santiago:2020:Weakly} have proposed a notion of approximate submodularity that extends to the non-monotone case, as well.

\section{Preliminaries and Key Definitions}
\label{sec:weak-subm-from}
Throughout the remainder of the paper, all set functions $f : 2^X\rightarrow \RR_{\geq 0}$ that we consider will be \emph{monotone}, satisfying $f(B) \geq f(A)$ for all $A \subseteq B$. Because we are focusing on maximization problems, we will further assume without loss of generality that our objective functions are \emph{normalized}, so $f(\emptyset) = 0$. When there is no risk of confusion, we will use the shorthands $A+e$ for $A \cup \{e\}$ and $A - e$ for $A \setminus \{e\}$. We use the notation $f(e | A) \triangleq f(A + e) - f(A)$ for the marginal gain obtained in $f$ when adding an element $e \not \in A$ to $A$.

The problems that we consider will be constrained by an arbitrary matroid $\cM=(X,\cI)$. Here $\cI \subseteq 2^X$ is a family of \emph{independent sets}, satisfying $\emptyset \in \cI$, $A \subseteq B \in \cI \Rightarrow A \in \cI$, and for all $A,B \in \cI$ with $|A| < |B|$, there exists $e \in B \setminus A$ such that $A+e \in \cI$. The maximal independent sets of $\cI$ are called \emph{bases} of $\cM$, and the last condition implies that they all have the same size, called the \emph{rank} of $\cM$, which we typically denote by $k$. Our goal will be to find some $S \in \cI$ that maximizes the objective $f$. Since $f$ is monotone, we can assume without loss of generality that $S$ and the maximizer of $f$ is a base of $\cM$. Throughout, we will make use of the following standard result (in fact, this is the only property of matroids that we will use for our analyses):
\begin{proposition}\label{prop:exchange}
Let $\cM=(X,\cI)$ be a matroid. Then for any pair of
bases $A,B$ of $\cM$, there exists a bijection $\pi : A \to B$ so that $A - a + \pi(a) \in \cI$ for all $a \in A$.
\end{proposition}

A set function $f : 2^X \to \posreals$ is \emph{submodular} if and only if $f(e | A) \geq f(e | B)$ for all $A \subseteq B \subseteq X$ and $e \not \in B$. It can be shown that this is equivalent to requiring that $f(B) - f(A) \leq \sum_{e \in B \setminus A}f(e | A)$
for any $A \subseteq B \subseteq X$. The definition of \emph{weak-submodularity}\footnote{We note that the definition given here, which is also used in \cite{Bian:2017:Guarantees,Chen:2018:Weakly,Elenberg:2017:StreamWeak,DBLP:conf/icml/HarshawFWK19,Santiago:2020:Weakly}, is slightly adapted from the original definition given in \cite{Das:2011:Submodular}.} relaxes this inequality by requiring:
\begin{equation}
\label{eq:gamma-submod}
\gamma \cdot \left(f(B)  - f(A)\right) \leq \sum_{e \in B \setminus A}f(e | A)\,,
\end{equation}
for any set $A \subseteq B$, where $\gamma \in [0,1]$ is called the submodularity ratio of $f$. Note that when $\gamma \geq 1$, \eqref{eq:gamma-submod}  iff the function $f$ is submodular. 
Submodularity can equivalently be characterized by
$f(B) - f(A) \geq \sum_{e \in B \setminus A} f(e | B - e)$
for all $A \subseteq B \subseteq X$. Thus, another natural approach is to consider functions that satisfy:
\begin{equation}
\label{eq:beta-submod}
\beta \cdot \left(f(B) - f(A)\right) \geq \sum_{e \in B \setminus A}f(e | B - e)
\end{equation}
for some $\beta \geq 1$. We call this property \emph{$\beta$-weak submodularity from above} to distinguish it from~\eqref{eq:gamma-submod}, which we will now refer to as \emph{$\gamma$-weak submodularity from below}.
Here, $\beta \leq 1$ if and only if $f$ is submodular.
Note that by monotonicity, $f(B) - f(A) \geq f(B) - f(B - e) \geq 0$ for any $e \in B \setminus A$ and so every monotone function $f$ satisfies~\eqref{eq:beta-submod} for $\beta = \card{B}$. We say that a set function $f$ is \emph{$(\gamma,\beta)$}\emph{-weakly submodular} if it is $\gamma$-weakly submodular from below and $\beta$-weakly submodular from above (i.e.\ it satisfies both \eqref{eq:gamma-submod} and \eqref{eq:beta-submod}).

\section{Subset Selection}
\label{sec:R2}
We now turn to the subset selection problem. Let $\cM = (X, \cI)$ be a matroid. Let $Z$ be a target random variable we wish to predict, and let $\cX = \{X_1,\ldots,X_n\}$ be a set of $n$ predictor variables (where here and throughout this section we use calligraphic letters to denote sets of random variables to avoid confusion). We suppose that $Z$ and all $X_i$ have been normalized to have mean 0 and variance 1, and let $C_{\cX}$ be the  $n \times n$ covariance matrix for the variables $X_i$. Our goal is to find a set $\calS \subseteq \cX$, that is independent in some given matroid over $\cX$ and gives the best linear predictor for $Z$. In other words, we want to solve the following optimization problem:
$$ \argmax_{\calS \in \cI} R^2_{Z, \calS} = \argmax_{\calS \in \cI} (\var(Z) - \esp{(Z - Z_{\calS})^2})/\var(Z), $$
where $R^2$ is a measure of fitness of the linear predictor using the \emph{squared multiple correlation}, and   $Z_{\calS} = \sum_{X_i \in \calS} \alpha_i X_i$ is the linear predictor over $\calS$ which optimally minimizes the mean square prediction error for $Z$.
The coefficients of this best linear predictor are given by $\bm{\alpha} = C_{\calS}^{-1}\vb_{Z,\calS}$, where $C_{\calS}$ is the principle submatrix of $C_{\cX}$ corresponding to variables in $\calS$, and $\vb_{Z,\calS}$ is a vector of covariances between $Z$, and each $X_i \in \calS$, i.e. $(C_{\calS})_{i,j} = \cov(X_i, X_j)$ and $(\vb_{Z,\calS})_{i} = \cov(X_i,Z)$. Therefore, if we let $X_{\calS}$ denote the corresponding vector of random variables in $\calS$, the best linear predictor can be written as: $Z_{\calS} = X_{\calS}\transp C_{\calS}^{-1}\vb_{Z,\calS}$.
Because $Z$ has unit variance, the objective simplifies to $R^2_{Z, \calS} = 1 - \esp{(Z-Z_{\calS})^2}$, and so
the $R^2$ objective can be regarded as a measure of the fraction of variance of $Z$ that is explained by $\calS$.
In addition, we can define the \emph{residual} of $Z$ with respect to this predictor as the random variable $\res(Z,\calS) = Z - Z_{\calS} = Z - X_{\calS}\transp C_{\calS}^{-1}\vb_{Z,\calS}$. Therefore, $R^2_{Z, \calS} = 1 - \var(\res(Z,\calS)) = \vb^T_{Z, \calS} C_{\calS}^{-1}\vb_{Z, \calS}$.

Das and Kempe \cite{DBLP:journals/jmlr/DasK18} show that the $R^2$ objective satisfies~\eqref{eq:gamma-submod} for all $\cA \subseteq \cB \subseteq \cX$ with $\gamma \geq \lambda_{\min}(C_{\cX}, \card{\cB}) \geq  \lmin(C_{\cX})$, where $\lmin(C_{\cX})$ is the smallest eigenvalue of $C_{\cX}$ and $\lmin(C_{\cX}, \card{\cB})$ is the smallest \emph{$\card{\cB}$-sparse eigenvalue of $C_{\cX}$}. In this section we derive an analogous result.
\begin{theorem}\label{thm:beta-R2}
    For any $\cB \subseteq \cX$, the $R^2$ objective satisfies \eqref{eq:beta-submod} with $\beta \leq \frac{1}{\lmin(C_{\cX}, \card{\cB})} \leq \frac{1}{\lmin(C_{\cX})}$.
    \end{theorem}
    Combined with existing bounds for $\gamma$, it shows that the $R^2$ objective is $(\gamma,1/\gamma)$-weakly submodular for $\gamma = \lmin(C_{\cX})$. To analyze the greedy algorithm, it suffices to let $\cB \setminus \cA$ and $\cA$ both be subsets containing at most $k$ variables, so $\card{\cB} = 2k$, which can lead to tighter bounds on $\gamma$. In fact, for both the algorithms we consider in the next two sections, it suffices to consider sets $B$ of size $k$ in \eqref{eq:beta-submod} and so in practice the tighter bound of $\beta \leq 1/\lmin(C_{\cX},k)$ holds.

    First, we consider the following small example that illustrates that the \emph{element-wise} bounds on $\game$ (inverse curvature, DR ratio, or generic submodularity ratio) are in general not bounded by $\lmin(C_\cX)$. In fact, we may have $\game = 0$ (and so approximation bounds based on $\game$ fail) even when $\lmin(C_\cX)$ is bounded away from 0.
\begin{example}
  Let $Z$,$X_1$,$X_2$ be random variables with unit variance and zero mean. Suppose that $X_1$ is uncorrelated with $Z$, and $X_2 = (Z+X_1)/\sqrt{2}$. Then, $\cov(X_1,Z) = 0$ and $\cov(X_1,X_2) = \cov(X_2,Z) = 1/\sqrt{2}$. Let $f(\mathcal{S}) = R^2_{Z,\mathcal{S}}$. Then, it can be verified that $f(X_1 | \emptyset) = 0$ and $f(X_1 | \{X_2\}) = 1/2$. Thus $f(X_1 | \emptyset) \geq \game \cdot f(X_1 | \{X_2\})$ is satisfied only for $\game = 0$. However, $\lmin(C_{\cX}^{-1}) = 1 - 1/\sqrt{2}$ and, in fact, explicitly computing $\gamma$ gives $\gamma = 1/2$.
    \label{ex:suppressor}
 \end{example}

Next, we turn to the proof of Theorem \ref{thm:beta-R2}. In order to prove our bounds, we will use the following facts stated in~\cite{DBLP:journals/jmlr/DasK18}:
\begin{lemma}
    Given two sets of random variables $\calS = \{X_1, \ldots, X_n\},$ and $ \cA$, and a random variable $Z$ we have: $\res(Z, \cA \cup \calS) = \res(\res(Z, \cA), \{\res(X_i, \cA)\}_{X_i \in \calS})$.
    \label{lem:2.3-das}
\end{lemma}
\begin{lemma}
    Given two sets of random variables $\calS = \{X_1, \ldots, X_n\},$ and $ \cA$, and a random variable $Z$ we have:
    $R^2_{Z, \cA \cup \calS} = R^2_{Z, \cA} + R^2_{Z, \{\res(X_i, \cA)\}_{X_i \in \calS}}$.
    \label{lem:2.4-das}
\end{lemma}

We define following quantities, which we use for the rest of the section. Let $\cA, \cB$ be some fixed sets of random variables with $\cA \subseteq \cB$. Let $\cT = \cB \setminus \cA$ and suppose without loss of generality that $\cT = \{X_1, \ldots, X_t\}$. For each $X_i \in \cT$, let $\hX_i = \res(X_i,\cA)$ and suppose further that each $\hX_i$ has been renormalized to have unit variance. Let $\hcT = \{\hX_i,\ldots,\hX_t\}$, $\hC$  to be the covariance matrix for $\hcT$, and $\hb$ to be the vector of covariances between $Z$ and each $\hX_i \in \hcT$.
We fix a single random variable $X_i$. For ease of notation, in the next two Lemmas we assume without loss of generality that $\hC$ and $\hb$ have been permuted so that $X_i$ corresponds to the last row and column of $\hC$. Then, we define $\cT_{-i} = \cT \setminus \{X_i\}$, $\hcT_{-i} = \hcT \setminus \{\hX_i\}$, and let $\hX_{-i}$ denote the vector containing the variables of $\hcT_{-i}$ (ordered as in $\hC$ and $\hb$). Similarly, let $\hC_{-i}$ be the principle submatrix of $\hC$ obtained by excluding the row and column corresponding to $\hX_i$ (i.e., the last row and column), and $\hb_{-i}$ be the vector obtained from $\hb$ by excluding the entry for $\hX_i$ (i.e., the last entry). Finally, we let $\vu_i$ be the vector of covariances between $\hX_i$ and each $\hX_j \in \hcT_{-i}$. Note that $\vu_i$ corresponds to the last column of $\hC$ with its last entry (corresponding to $\var(\hX_i)$) removed.
We begin by computing the loss in $R^2_{Z,\cB}$ when removing $X_i$ from $\cB$:
\begin{lemma}
$R^2_{Z, \cB} - R^2_{Z, \cB \setminus \{X_i\}} = \cov(Z,\res(\hX_i,\hcT_{-i}))^2/\var(\res(\hX_i,\hcT_{-i})) = \hb^TH_i\hb / s_i$, where $H_i = \begin{pmatrix}
\hC_{-i}^{-1} \vu_i \vu_i^T \hC_{-i}^{-1} & -\hC_{-i}^{-1} \vu_i \\
            - \vu_i^T \hC_{-i}^{-1} & 1
        \end{pmatrix}$ and $s_i = 1 - \vu_i^T \hC_{-i}^{-1} \vu_i$.
    \label{lem:R2-num1}
\end{lemma}


\begin{proof}[Proof of Lemma \ref{lem:R2-num1}]
    Note that $\cB \setminus \{X_i\} = \cA \cup \cT_{-i}$. Thus, by Lemma \ref{lem:2.4-das} and Lemma \ref{lem:2.3-das}, respectively:
\begin{equation}
        R^2_{Z, \cB} - R^2_{Z, \cB \setminus \{X_i\}} = R^2_{Z, \res(X_i,\cA\cup\cT_{-i})} = R^2_{Z, \res(\res(X_i, \cA), \{\res(X_j,\cA)\}_{X_j \in \cT_{-i}}}.
\label{eq:eq:r2-basic}
     \end{equation}
Recall that each $\hX_j$ is obtained from $\res(X_j,\cA)$ by renormalization and that $\hcT_{-i} = \hcT \setminus \{\hX_i\} = \{\hX_j\}_{X_j \in \cT_{-i}}$. Thus, $\res(\hX_i,\hcT_{-i})$ is a rescaling of $\res(\res(X_i, \cA), \{\res(X_j,\cA)\}_{X_j \in \cT_{-i}})$. Since the $R^2$ objective is invariant under scaling of the predictor variables, \eqref{eq:eq:r2-basic} then implies that
    \begin{equation}
        R^2_{Z,\cB} - R^2_{Z, \cB \setminus \{X_i\}} = R^2_{Z,\res(\hX_i, \hcT_{-i})} = \cov(Z,\res(\hX_i,\hcT_{-i}))^2/\var(\res(\hX_i,\hcT_{-i}))\,,
        \label{eq:r2-cov-var}
    \end{equation}
where the last line follows directly from the definition of the $R^2$ objective.
It remains to express \eqref{eq:r2-cov-var} in terms of $\hC$, $\hb$ and $\vu$. By definition, $\res(\hX_i, \hcT_{-i}) = \hX_i - \hX_{-i}^T \hC_{-i}^{-1}\vu_{i}$. Hence,
    \begin{multline*}
        \var(\res(\hX_i, \hcT_{-i}))  = \var(\hX_i -  \hX_{-i}^T\hC_{-i}^{-1}\vu_i) \\
         = \expect[\hX_i^2] - 2 \expect[\hX_i\hX^T_{-i}]\hC_{-i}^{-1}\vu_i + \vu_i^T \hC_{-i}^{-1} \expect[\hX_{-i} \hX_{-i}^T]\hC_{-i}^{-1} \vu_i
         = 1 - \vu_i^T \hC_{-i}^{-1} \vu_i,
\end{multline*}
where the last equality follows from normalization of $\hX_i$, $\expect[\hX_{i}\hX_{-i}^T] = \vu_i^T$ and $\expect[\hX_{-i}\hX_{-i}^T] = \hC_{-i}$. Furthermore,
    \begin{multline*}
        \cov(Z, \res(\hX_i, \hcT_{-i}))^2  = \cov( Z, \hX_i -  \hX_{-i}^T\hC_{-i}^{-1}\vu_i)^2 = \lb \cov(Z, \hX_i) - \cov(Z, \hX_{-i}^T\hC_{-i}^{-1}\vu_i) \rb^2 \\
         = \lb \hat{b}_i - \hb_{-i}^T C_{-i}^{-1}\vu_i \rb^2
         = \hb^T\!
        \begin{pmatrix}
            \hC_{-i}^{-1} \vu_i \vu_i^T \hC_{-i}^{-1} & -\hC_{-i}^{-1} \vu_i \\
            - \vu_i^T \hC_{-i}^{-1} & 1
        \end{pmatrix}\!
        \hb.
    \end{multline*}
Substituting the above 2 expressions into \eqref{eq:r2-cov-var} completes the proof.
\end{proof}
In the next lemma we show that the previous lemma can be simplified for eigenvectors of $\hC^{-1}$.
\begin{lemma}
Let $(\lambda, \vv), (\mu, \vw)$ be any 2 eigenpairs of $\hC^{-1}$.
Then, $\vv^TH_i \vw = \lambda \mu s_i^2 v_i w_i$, where $H_i$ and $s_i$ are as defined in the statement of Lemma~\ref{lem:R2-num1}.
\label{lem:R2-eig}
\end{lemma}


\begin{proof}[Proof of Lemma \ref{lem:R2-eig}]
Applying the formula for block matrix inversion (Lemma~\ref{thm:block-inverse}) to $\hC^{-1}$, we have
    \begin{equation}
        \hC^{-1} =
        \begin{pmatrix}
            \hC_{-i} & \vu_i \\
            \vu_i^T & 1
        \end{pmatrix}^{-1}
        =
        \begin{pmatrix}
            \hC_{-i}^{-1} & 0 \\
            0 & 0
        \end{pmatrix}
         + \frac{1}{1 - \vu_i^T \hC_{-i}^{-1}\vu_i}
         \begin{pmatrix}
             \hC_{-i}^{-1} \vu_i \vu_i^T \hC_{-i}^{-1} & -\hC_{-i}^{-1} \vu_i \\
             - \vu_i^T \hC_{-i}^{-1} & 1
         \end{pmatrix}.
\label{eq:block-inverse-1}
    \end{equation}
Now, because $(\mu, \vw)$ is an eigenpair of $\hC^{-1}$, we must have $(\hC^{-1}\vw)_i = \mu w_i$. By~\eqref{eq:block-inverse-1}, this is equivalent to $(-\vu_i^T\hC_{-i}^{-1}\vw_{-i} + w_i)/s_i = \mu w_i$ (where, as usual, we let $\vw_{-i}$ be the vector obtained from $\vw$ by discarding its $i^{\textrm{th}}$ entry). Rearranging this equation gives $\vu_i^T\hC_{-i}^{-1}\vw_{-i} = w_i(1-\mu s_i)$. Since $\hC^{-1}$ is symmetric, the same argument implies that $\vv_{-i}^T \hC_{-i}^{-1} \vu_i = v_i(1-\lambda s_i)$. Thus,
\begin{align*}
\vv^T H_i \vw
            & = \vv_{-i}^T \hC_{-i}^{-1} \vu_i \vu_i^T \hC_{-i}^{-1} \vw_{-i} - w_i (\vv_{-i}^T \hC_{-i}^{-1} \vu_i) - v_i(\vu_i^T \hC_{-i}^{-1} \vw_{-i}
) + v_i w_i \\
             &  =  v_i w_i ( 1 - \lambda s_i)(1 - \mu s_i)\! - v_i w_i(1 - \lambda s_i)\! - v_i w_i(1 - \mu s_i)\! +\! v_i w_i \\
             & = v_i w_i \left((1- \lambda s_i)(1- \mu s_i) - (1- \lambda s_i) - (1 - \mu s_i) + 1\right) = \lambda \mu s_i^2 v_i w_i\,,
    \end{align*}
as claimed.
\end{proof}

We can now complete the proof of our main result from this section (Theorem~\ref{thm:beta-R2}).
\begin{proof}[Proof of Theorem \ref{thm:beta-R2}]
Let $\{\vv_1,\ldots,\vv_t\}$ be an eigenbasis of $\hC^{-1}$ with corresponding eigenvalues $\lambda_1, \ldots, \lambda_t$. Let $V$ be a matrix with columns given by these $\vv_i$. Since $\hC^{-1}$ is a symmetric positive semidefinite matrix, the matrix $V$ is orthonormal. Hence, we can write $\hb = V\vy$ for some vector $\vy$. By Lemma~\ref{lem:R2-num1}, $\hb^TH_i\hb = \cov(Z,\res(\hX_i,\hcT_{-i}))^2 \geq 0$ and $s_i = \var(\res(\hX_i,\hcT_{-i})) \leq 1$, for each $i = 1,\ldots,t$ and $R^2_{Z, \cB} - R^2_{Z, \cB \setminus \{X_i\}} = \hb^T H_i \hb/s_i \leq \hb^T H_i \hb / s_i^2 = \vy^T V^TH_iV \vy/s_i^2$.
Finally, by Lemma~\ref{lem:R2-eig}, $(V^T H_i V)_{\ell,m} = \lambda_\ell\lambda_m s_i^2 (\vv_\ell)_i(\vv_m)_i$. Thus, summing over all $i$ we have:
\begin{multline}
\sum_{i = 1}^tR^2_{Z, \cB} - R^2_{Z,\cB \setminus \{X_i\}}
\leq \sum_{i = 1}^t\sum_{\ell,m = 1}^t (y_\ell y_m\lambda_\ell\lambda_m)(\vv_\ell)_i(\vv_m)_i\,.
\\
= \sum_{\ell,m = 1}^t (y_\ell y_m\lambda_\ell\lambda_m)\sum_{i = 1}^t(\vv_\ell)_i(\vv_m)_i\,
= \sum_{i = 1}^t y_i^2 \lambda_i^2 \leq \lambda_{\max}(\hC^{-1})\sum_{i = 1}^t y_i^2\lambda_i\,,
\label{eq:R2-num-final}
\end{multline}
where the last equation follows from the orthonormality of the eigenvectors $\vv_i$. Moreover, by Lemma~\ref{lem:2.4-das}
    \begin{equation}
R^2_{Z,\cB} - R^2_{Z,\cA} = R^2_{Z,\hcT} = \hb^T\hC^{-1}\hb = \sum_{i=1}^t y_i^2 \lambda_i\,.
        \label{eq:R2-den-final}
\end{equation}
Combining \eqref{eq:R2-den-final} and \eqref{eq:R2-num-final}, we have $\sum_{i \in \calS}R^2_{Z,\cB} - R^2_{Z,\cB \setminus \{X_i\}} \leq \lambda_{\max}(\hC^{-1})[R^2_{Z,\cB} - R^2_{Z,\cA}]$ and so
inequality~\eqref{eq:beta-submod} is satisfied for $\beta = \lambda_{\max}(\hC^{-1}) = 1/\lambda_{\min}(\hC)$. It remains to bound $1/\lambda_{\min}(\hC)$ in terms of the eigenvalues of $C_{\cX}$. Recall that $\hC$ is a normalized covariance matrix for the random variables $\{\res(X_i,\cB \setminus \cA)\}_{X_i \in \cA}$. As shown in~\cite{DBLP:journals/jmlr/DasK18} (see Lemma~\ref{lem:das-kempe-eigenvalues} in Appendix~\ref{sec:basic-results-from} for a formal statement), this implies that $\lambda_{\min}(\hC) \geq \lambda_{\min}(C_{(\cB \setminus \cA) \cup \cA}) \geq \lambda_{\min}(C_{\cX}, \card{\cB}) \geq \lmin(C_{\cX})$. The claimed bound on $\beta$ then follows.
\end{proof}


\section{Improved Analysis of \textsc{ResidualRandomGreedy}}
\label{sec:impr-analys-rrg}
In this section, we show that we can derive stronger approximation guarantees for $(\gamma,\beta)$-weakly submodular functions by using the \textsc{ResidualRandomGreedy} algorithm considered in \cite{DBLP:conf/soda/BuchbinderFNS14,Chen:2018:Weakly}. Combined with the results from the previous section, this gives improved approximation bounds for the subset selection problem with an arbitrary matroid constraint $\cM$. The algorithm $\textsc{ResidualRandomGreedy}$ (shown in Algorithm \ref{alg:RRG}) proceeds over $k$ iterations. In iteration $i$, it greedily extends the current solution $S_{i-1}$ to a base $S_{i-1} \cup M_i$ of $\cM$ by selecting a set $M_i$ of the $k - |S_{i-1}| = k - i + 1$ elements with the largest marginal contribution with respect to the $S_{i-1}$. Then, it chooses an element $s_i$ uniformly at random from $M_i$ which is added to $S_{i-1}$ to obtain a new solution $S_{i}$. After $k$ iterations, the final set $S_k$ is returned.


\SetAlgoVlined
\begin{algorithm2e}
$S_0 \gets \emptyset$\;
\For{$i = 1, 2, \ldots, k$ }{
  $M_i \gets \argmax \lc \sum_{e \in T} f(e \mid S) : T \subseteq X, S \cup T \textrm{ is a base of $\cM$} \rc $\;
  $s_i \gets $ an element of $M_i$ chosen uniformly at random\;
  $S_{i} \gets S_{i-1} \cup \{s_i\}$\;
}
\Return{$S_k$}\;
\caption{$\textsc{ResidualRandomGreedy}(\cM,X,f)$}
\label{alg:RRG}
\end{algorithm2e}

	\begin{theorem}Suppose that $f : 2^X$ is $(\alpha,\beta)$-weakly submodular and $\cM = (X,\cI)$ is a matroid and let $\opt = \arg\max_{A \in \cI}f(A)$. Then, $\textsc{ResidualRandomGreedy}(\cM,X,f)$ returns a solution $S \in \cI$ that satisfies $\esp{f(S)} \geq \frac{\gamma}{\gamma+\beta}\cdot \fopt$.
		\label{thm:RRG}
	\end{theorem}
\begin{proof}[Proof of Theorem \ref{thm:RRG}] We begin by introducing some auxiliary sets used in the analysis. For each $i = 0, 1, \cdots, k$, we let $\opt_i$ to be a subset of $\opt$ of size $k - i$ such that $S_i \cup \opt_i$ is a base of $\cM$, as follows. Let $\opt_0 = \opt$. For each $i \geq 1$, suppose that $S_{i-1} \cup \opt_{i-1}$ is a base and consider the bijection $\pi_i: S_{i-1} \cup M_i \rightarrow S_{i-1} \cup \opt_{i-1}$ guaranteed by Proposition \ref{prop:exchange}. We set $\opt_i = \opt_{i-1} - \pi_i(s_i)$. Then, $S_i \cup \opt_i = S_{i-1} \cup \opt_{i-1} + s_i - \pi_i(s_i)$ is a base, as required.  Moreover, note the choice of $\pi_{i}$ is independent of the random choice $s_i$, which implies that $\pi_i(s_i)$ is an element of $\opt_{i-1}$ chosen uniformly at random. 
Let $\cE$ be the event which fixes the random decisions of the algorithm up to iteration $i-1$. Conditioned on $\cE$, we have:
		\begin{multline}
			\esp{f(S_i) - f(S_{i-1})} = \frac{1}{\card{M_i}} \sum_{e \in M_i} f(e | S_{i-1})
			= \frac{1}{k-i+1} \sum_{e \in M_i} f(e | S_{i-1}) \\
			 \geq \frac{1}{k-i+1} \sum_{e \in \opt_{i-1}} f(e | S_{i-1})
			\geq \frac{\gamma}{k-i+1} \lb f(\opt_{i-1} \cup S_{i-1}) - f(S_{i-1}) \rb \label{eq:S_i - S_{i-1}}.
		\end{multline}
		Here, the third inequality follows the fact that $S_{i-1} \cup \opt_{i-1}$ is a base and so $\opt_{i-1}$ is a candidate for $M_i$. The fourth inequality follows from \eqref{eq:gamma-submod} since $f$ is $(\gamma,\beta)$-weakly submodular. Similarly, \eqref{eq:beta-submod} together with the fact that $\pi_i(s_i)$ is a uniformly random element of $\opt_{i-1}$ implies
		\begin{align}
			\frac{1}{k-i+1} \lb f(\opt_{i-1} \cup S_{i-1}) - f(S_{i-1}) \rb & \geq \frac{\beta^{-1}}{k-i+1}  \sum_{e \in \opt_{i-1}} f(e | \opt_{i-1} \cup S_{i-1} - e), \nonumber\\
			& = \beta^{-1} \cdot \esp{f(\pi_i(s_i) | \opt_{i-1} \cup S_{i-1} - \pi_i(s_i)} \label{eq:apply_beta}.
		\end{align}
		We can bound the expected decrease in $f(\opt_i \cup S_i)$ in iteration $i$ as:
		\begin{align}
\expect[&f(\opt_i \cup S_i) - f(\opt_{i-1} \cup S_{i-1})] \nonumber \\ 
			&= \esp{f(\opt_{i-1} \cup S_{i-1} + s_i - \pi_i(s_i)) - f(\opt_{i-1} \cup S_{i-1})} \nonumber \\
			&= \esp{ f(s_i | \opt_{i-1} \cup S_{i-1} - \pi_i(s_i)) - f(\pi_i(s_i) | \opt_{i-1} \cup S_{i-1} - \pi_i(s_i))} \nonumber\\
			&\geq - \esp{f(\pi_i(s_i) | \opt_{i-1} \cup S_{i-1} - \pi_i(s_i))} \label{eq:telescop_RRG},
		\end{align}
		where the inequality follows by monotonicity of $f$. Thus
		\begin{align}
			\esp{f(S_i) - f(S_{i-1})} & \geq \frac{\gamma}{\beta} \esp{f(\pi_i(s_i) | \opt_{i-1} \cup S_{i-1} - \pi_i(s_i))} \nonumber \\
			& \geq \frac{\gamma}{\beta} \esp{ f(\opt_{i-1} \cup S_{i-1}) - f(\opt_i \cup S_i) },
\label{eq:RRG-final}
		\end{align}
		where the first inequality follows by combining \eqref{eq:S_i - S_{i-1}} and \eqref{eq:apply_beta} and the second by \eqref{eq:telescop_RRG}.
		
Removing the conditioning on $\cE$ and summing the inequalities \eqref{eq:RRG-final} for $i = 1, \cdots, k$, gives $\expect[f(S_k) - f(S_0)] \geq \frac{\gamma}{\beta}\expect[f(S_0 \cup \opt_0) - f(S_k \cup \opt_k)]$. The claim then follows by observing that $S_0 = \emptyset$, $S_k = S$,  $\opt_0 \cup S_0 = \opt$ and $\opt_{k} \cup S_k = S_k$ and so $(1 + \frac{\beta}{\gamma})\expect[f(S)] \geq f(\opt)$.
\end{proof}


\section{Distorted Local Search}
\label{sec:non-oblivious-local}
Here, we present an algorithm for $(\gamma,\beta)$-weakly submodular functions with a guarantee that smoothly approaches the optimal value of $(1-1/e)$ as $\gamma,\beta \to 1$. The algorithm (Algorithm~\ref{alg:dist-local-search-simple}), is a local search routine that attempts to swap a single element into the current solution if and only if it improves the following auxiliary potential function parameterized by $\phi \in \posreals$, which we will set appropriately depending on $\gamma$ and $\beta$:
\begin{equation*}
\label{eq:g-def}
g_\phi(A) =  \int_0^1\!\! \frac{\phi e^{\phi p}}{e^{\phi} - 1}\sum_{B \subseteq A}p^{|B|-1}(1-p)^{|A|-|B|}f(B)\,dp =\sum_{B \subseteq A}\mc{\phi}{|A|-1}{|B|-1},
f(B)
\end{equation*}
where we define
\begin{equation*}
\mc{\phi}{a}{b} \triangleq \int_0^1 \phi e^{\phi p}p^{b}(1-p)^{a - b}/(e^{\phi} - 1)\,dp.
\end{equation*}
\SetAlgoVlined
\begin{algorithm2e}[t]
      Suppose that $f$ is $(\gamma,\beta)$-weakly submodular and let $\phi = \gamma^2 + \beta(1-\gamma)$\;
      $A \gets $ an arbitrary base of $\cM$\;
        \While{$\exists a \in S, b \in X \setminus S$ with $S - a + b \in \cI$ and $g_\phi(A - a + b) > g_\phi(A)$}
        {
          $A \gets S - a + b$\;
        }
        \Return{$A$}\;
    \caption{$\textsc{DistortedLocalSearch}(\cM, X, f)$\label{alg:dist-local-search-simple}}
\end{algorithm2e}
In the analysis of~\cite{Filmus:2014}, it is shown that if $f$ is submodular, its associated potential $g$ is as well, and this plays a crucial role in the analysis. Here, however, $f$ is only \emph{weakly} submodular, which means we must carry out an alternative analysis to bound the quality of a local optimum for $g_\phi$. Our analysis will rely on the following properties of the coefficients $\mc{\phi}{a}{b}$ (see Appendix~\ref{sec:prop-coeff-mcph} for a full proof of each):
\begin{restatable}{mylemma}{coeffprops}\label{lem:coeff-props}
For any $\phi > 0$, the coefficients $\mc{\phi}{a}{b}$ satisfy the following:
\begin{enumerate}
\item $g_\phi(e|A) = \sum_{B \subseteq A}\mc{\phi}{|A|}{|B|}f(e|B)$, for any $A \subseteq X$ and $e \not \in A$.
\item $\sum_{B \subseteq A} \mc{\phi}{|A|}{|B|} = 1$, for all $A \subseteq X$.
\item $\mc{\phi}{a}{b} = \mc{\phi}{a+1}{b+1} + \mc{\phi}{a+1}{b}$ for all $0
 \leq b \leq a$.
\item $\phi \mc{\phi}{a}{b} = -b \mc{\phi}{a-1}{b-1} + (a-b) \mc{\phi}{a-1}{b} + (\phi/(e^{\phi - 1})) \bm{1}_{b = 0} + (\phi e^\phi/(e^{\phi} - 1)) \bm{1}_{b = a}$,\\ for all $a > 0$ and $0 \leq b \leq a$.
\end{enumerate}
\end{restatable}

In order to analyze the performance of Algorithm~\ref{alg:dist-local-search-simple}, we consider now two arbitrary bases $A$ and $O$ of the given matroid $\cM$. We index the elements $a_i \in A$ and $o_i \in O$ according to the bijection $\pi : A \to O$ guaranteed by Proposition~\ref{prop:exchange} so that $A - a_i + o_i$ is a base for all $1 \leq i \leq |A|$. Our main theorem is the following:
\begin{theorem}
\label{thm:loc-opt-main}
Suppose that $f$ is $(\gamma,\beta)$-weakly submodular and let $\phi = \phi(\gamma,\beta) \triangleq \gamma^2 + \beta(1 - \gamma)$. Then, for any bases $A,O$ of a matroid $\cM$, $\frac{\phi e^\phi}{e^{\phi} - 1}f(A) \geq \gamma^2 f(O) + \sum_{i = 1}^{|A|}[g_{\phi}(A) - g_{\phi}(A - a_i + o_i)]$.
\end{theorem}
Note that the base $A$ ultimately returned by Algorithm~\ref{alg:dist-local-search-simple}, necessarily has $g_\phi(A) - g_\phi(A - a_i + o_i) \leq 0$ for all $i$, so Theorem~\ref{thm:loc-opt-main} immediately implies that $f(A) \geq \gamma^2\frac{(1- e^{-\phi})}{\phi}f(O)$, where $\phi = \gamma^2 + \beta(1-\gamma)$.

To prove Theorem~\ref{thm:loc-opt-main}, we first note that $g_{\phi}(A - a_i + o_i) - g(A) = g(o_i | A - a_i) - g(a_i | A - a_i)$ and so 
\begin{equation}
\label{eq:g-loc-opt-1}
\sum_{i = 1}^{|A|}g(a_i | A - a_i) = \sum_{i = 1}^{|A|}[g(A) - g(A - a_i + o_i)] + \sum_{i = 1}^{|A|}g(o_i | A - a_i).
\end{equation}
The final term in \eqref{eq:g-loc-opt-1} can now be bounded as follows:
\begin{lemma}\label{lem:g-local-opt-1}
Suppose that $f$ is $(\gamma, \beta)$-weakly submodular, and let $A,O \subseteq X$ with $A = \{a_1,\ldots,a_{|A|}\}$ and $O = \{o_1,\ldots,o_{|A|}\}$ (so $|A| = |O|$). Then,
\begin{equation*}
\sum_{i=1}^{|A|}g(o_i | A - a_i) \geq \gamma^2 f(O) - \left(\gamma^2 + \beta(1 - \gamma)\right)\sum_{B \subseteq A}\mc{\phi}{|A|}{|B|}f(B).
\end{equation*}
\end{lemma}
\begin{proof}[Proof of Lemma \ref{lem:g-local-opt-1}]
By parts 1 and 3 of Lemma~\ref{lem:coeff-props}, we have
\begin{equation}
g_{\phi}(o_i | A - a_i) = \sum_{\mathclap{B \subseteq A-a_i}} \mc{\phi}{|A|-1}{|B|}f(o_i | B) =\sum_{\mathclap{B \subseteq A - a_i}}[\mc{\phi}{|A|}{|B|+1}f(o_i | B) + \mc{\phi}{|A|}{|B|}f(o_i | B)].
\label{eq:1}
\end{equation}
Since $f$ is $\gamma$-weakly submodular from below
\[f(o_i | B) + f(a_i | B) \geq \gamma f(B \cup \{o_i, a_i\}) - \gamma f(B)
= \gamma f(o_i | B + a_i) + \gamma f(a_i | B),\]
and so $f(o_i | B) \geq \gamma f(o_i | B+a_i) - (1-\gamma)f(a_i | B)$. Thus, the right-hand side of~\eqref{eq:1} is at least
\begin{equation}
\label{eq:loc-upper-bound-1}
\sum_{\mathclap{B \subseteq A - a_i}}\mc{\phi}{|A|}{|B|+1}\left[\gamma f(o_i | B+a_i) - (1-\gamma)f(a_i | B)\right] + \mc{\phi}{|A|}{|B|}f(o_i | B) = P + Q,
\end{equation}
where
\begin{align*}
P &= \gamma \sum_{\mathclap{B \subseteq A-a_i}}\left[\mc{\phi}{|A|}{|B|+1} f(o_i | B+a_i) + \mc{\phi}{|A|}{|B|}f(o_i | B)\right] = \gamma \sum_{B \subseteq A}\mc{\phi}{|A|}{|B|} f(o_i | B) \\
Q &= (1-\gamma)\sum_{\mathclap{B \subseteq A-a_i}}\left[\mc{\phi}{|A|}{|B|}f(o_i | B) - \mc{\phi}{|A|}{|B|+1}f(a_i | B)\right]\! \geq\! -(1-\gamma)\sum_{\mathclap{B \subseteq A-a_i}} \mc{\phi}{|A|}{|B|+1}f(a_i | B).
\end{align*}
In the first equation, we have used that for each set $T \subseteq A$, $f(o_i|T)$ appears in the right-hand summation exactly once: if $a_i \in T$ it appears as $T = B+a_i$ with coefficient $\mc{\phi}{|A|}{|B|+1} = \mc{\phi}{|A|}{|T|}$ and if $a_i \not\in T$ it appears as $T=B$ with coefficient $\mc{\phi}{|A|}{|B|} = \mc{\phi}{|A|}{|T|}$. Summing \eqref{eq:loc-upper-bound-1} over each $a_i \in A$ we then have
\begin{equation}
\label{eq:loc-bound-1}
\sum_{i = 1}^{|A|} \ g_{\phi}(o_i | A-a_i) \geq
\gamma \sum_{i = 1}^{|A|} \sum_{B \subseteq A} \mc{\phi}{|A|}{|B|}f(o_i | B) - (1-\gamma)\sum_{i = 1}^{|A|} \sum_{B \subseteq A - a_i} \mc{\phi}{|A|}{|B|+1}f(a_i | B) .
\end{equation}
Since $f$ is $\gamma$-weakly submodular from below and monotone,
\begin{multline*}
\gamma \sum_{i = 1}^{|A|} \sum_{B \subseteq A} \mc{\phi}{|A|}{|B|} f(o_i | B)
= \gamma \sum_{B \subseteq A} \sum_{i = 1}^{|A|} \mc{\phi}{|A|}{|B|} f(o_i | B)
\geq
\gamma^2 \sum_{B \subseteq A} \mc{\phi}{|A|}{|B|}[f(O \cup B) - f(B)] \\
\geq
\gamma^2 \sum_{B \subseteq A} \mc{\phi}{|A|}{|B|}[f(O) - f(B)]
= \gamma^2 f(O) - \gamma^2 \sum_{B \subseteq A}\mc{\phi}{|A|}{|B|}f(B),
\end{multline*}
where the last equation follows from part 2 of Lemma~\ref{lem:coeff-props}. Similarly, since $f$ is $\beta$-weakly submodular from above:
\begin{multline*}
(1-\gamma)\sum_{i = 1}^{|A|}\sum_{B \subseteq A-a_i}\!\!\!\mc{\phi}{|A|}{|B|+1}(f(B+a_i) - f(B))
= (1-\gamma)\sum_{T \subseteq A}\sum_{i = 1}^{|A|}\mc{\phi}{|A|}{|T|}(f(T) - f(T-a_i)) \\
\leq \beta(1-\gamma)\sum_{T \subseteq A} \mc{\phi}{|A|}{|T|}[f(T) - f(\emptyset)] =
\beta(1-\gamma)\sum_{B \subseteq A} \mc{\phi}{|A|}{|B|}f(B),
\end{multline*}
where the first equation can be verified by substituting $B = T - a_i$ for each $a_i \in T$ and noting that $|T| = |B|+1$, and the last equation simply follows from $f(\emptyset) = 0$ and renaming $T$ to $B$. Using the two previous inequalities to bound the right-hand side of~\eqref{eq:loc-bound-1}, then gives the claimed result.
\end{proof}

\begin{proof}[Proof of Theorem \ref{thm:loc-opt-main}]
Applying Lemma~\ref{lem:g-local-opt-1} to the last term in~\eqref{eq:g-loc-opt-1} and rearranging gives:
\begin{equation}
\sum_{i = 1}^{|A|} g_{\phi}(a_i| A - a_i) + \left(\gamma^2 + \beta(1 - \gamma)\right)\sum_{\mathclap{B \subseteq A}} \mc{\phi}{|A|}{|B|}f(B) \geq \gamma^2 f(O) + \sum_{i = 1}^{|A|}[g_{\phi}(A) - g_{\phi}(A - a_i + o_i)].
\label{eq:g-loc-opt-2}
\end{equation}
From part 1 of Lemma~\ref{lem:coeff-props},
\begin{align*}
\sum_{i = 1}^{|A|} g_{\phi}(a_i | A - a_i)
&= \sum_{i = 1}^{|A|} \sum_{B \subseteq A - a_i} \mc{\phi}{|A|-1}{|B|}(f(B+a_i) - f(B)) \\
&= \sum_{T \subseteq A} |T|\mc{\phi}{|A|-1}{|T|-1}f(T) - (|A| - |T|)\mc{\phi}{|A|-1}{|T|}f(T),
\end{align*}
where the last equation follows from the fact that each $T \subseteq A$ appears once as $T = B + a_i$ for each $a_i \in T$ (in which case it has coefficient $\mc{\phi}{|A|-1}{|B|}=\mc{\phi}{|A| - 1}{ |T|-1}$) and once as $T = B$ for each $a_i \not\in T$ (in which case it has coefficient $\mc{\phi}{|A|-1}{|B|} = \mc{\phi}{|A|-1}{|T|}$). Thus, we can rewrite~\eqref{eq:g-loc-opt-2} as:
\begin{multline}
\label{eq:g-loc-opt-main}
\sum_{B \subseteq A}\left(|B|\mc{\phi}{|A|-1}{|B|-1} - (|A| - |B|)\mc{\phi}{|A|-1}{|B|} + \left(\gamma^2 + \beta(1-\gamma)\right)\mc{\phi}{|A|}{|B|}\right)f(B) \\ \geq \gamma^2 f(O) + \sum_{i = 1}^{|A|}[g_{\phi}(A) - g_{\phi}(A - a_i + o_i)].
\end{multline}
Since $\phi = \gamma^2 + \beta(1-\gamma)$, the recurrence in part 4 of Lemma~\ref{lem:coeff-props} implies that the left-hand side vanishes for all $B$ except $B = \emptyset$, in which case it is $\frac{\phi}{e^{\phi}-1}f(\emptyset) = 0$ or $B = A$, in which case it is $\frac{\phi e^\phi}{e^{\phi} - 1}f(A)$. The theorem then follows.
\end{proof}
There are several further issues that must be addressed in order to convert Algorithm~\ref{alg:dist-local-search-simple} to a general, polynomial-time algorithm. First, we cannot compute $g_\phi(A)$ directly, as it depends on the values $f(A)$ for all subsets of $A$. In Appendix~\ref{sec:effic-estim-g_phi} we show that we can efficiently estimate $g_\phi$ via simple sampling procedure.
To bound the number of improvements made, we can instead require that each improvement makes a $(1+\epsilon)$ increase in $g_\phi$. Then at termination, we will instead have $\sum_{i = 1}^{|A|} g_\phi(A) - g_\phi(A - a_i + o_i)] \leq |A|\epsilon g_{\phi}(A)$. In order to bound the resulting loss in our guarantee we must bound the value $g_\phi(A)$ in terms of $f(A)$, which we accomplish in Appendix~\ref{sec:bounding-value-g}. Finally, we address the fact that $\gamma$ and $\beta$ may not be known and so we cannot set $\phi$ a priori. We show that by initializing the algorithm with a solution produced by \textsc{ResidualRandomGreedy}, we can bound the range of values for $\phi$ that must be considered to obtain our guarantee. It then suffices to enumerate guesses for $\phi$ from this range. In Appendix~\ref{sec:bound-sens-g_phi} we show that small changes in $\phi$ result in small changes to $g_\phi(A)$, and so by initializing the run for each subsequent guess of $\phi$ with the solution produced for the previous guess, we can amortize the total number of improvements (and work) required across all guesses. The final algorithm, presented in Appendix~\ref{sec:our-final-algorithm}, has the same guarantee as Algorithm~\ref{alg:dist-local-search-simple} minus a small $\cO(\epsilon)$ term, and requires $\tilde{\cO}(nk^4\epsilon^{-3})$ evaluations of $f$. Thus, we have the following:
\begin{restatable}{theorem}{dlsmain}\label{thm:distorted-ls-main}
Let $\cM=(X,\cI)$ be a matroid, $f : 2^X \to \posreals$ be a $(\gamma,\beta)$-weakly submodular function, and $\epsilon > 0$. Then, there is a randomized algorithm that with  probability $1 - o(1)$ returns a set $S$ satisfying $f(S) \geq \left(\frac{\gamma^2 (1 - e^{-\phi(\gamma,\beta)})}{\phi(\gamma,\beta)} - \cO(\epsilon)\right)f(O)$ for any solution $O \in \cI$, where $\phi(\gamma,\beta) = \gamma^2 + \beta(1-\gamma)$. The algorithm runs in time $\tilde{\cO}(nk^4\e^{-3})$.
\end{restatable}
In both the subset selection problem and Bayesian $A$-optimal design (considered in Appendix~\ref{sec:optimal-design}), we can derive an upper bound for $\beta$ matching existing spectral bounds on $\gamma^{-1}$. Then we have $\phi(\gamma,\beta) = \gamma^2 + \frac{1}{\gamma} - 1$ and so for both problems we obtain a guarantee of $\gamma^2 \cdot \frac{1 -e^{-(\gamma^2+\gamma^{-1}-1)}}{\gamma^2 + \gamma^{-1} - 1} - \cO(\e)$. In particular, as $\gamma$ tends to 1 (and so $f$ becomes closer to submodular) our guarantee approaches $1 - e^{-1} - \cO(\e)$, matching (up to $\cO(\e)$) the optimal guarantee for submodular functions. 
Surprisingly, we show that this relationship between $\gamma^{-1}$ and $\beta$ does not hold in general. In Appendix~\ref{sec:how-large-can}, we prove the following:
\begin{restatable}{theorem}{betalarge}\label{lem:worst-case-beta}
    For any $\gamma > 0$ and $k > 0$ there exists a function on a ground set of size $k$ that is $\gamma$-weakly submodular from below but not $\beta$-weakly submodular from above for any $\beta <  \binom{k  - \gamma}{k-1} = \bigT{k^{1-\gamma}}$.
\end{restatable}
Note that as $\gamma \rightarrow 0$ (respectively $\gamma \rightarrow 1$), we have $\beta \rightarrow k$ (respectively $\beta \rightarrow 1$). Hence, we recover the trivial upper bound for monotone set functions when $\gamma \rightarrow 0$ and as $\gamma \to 1$, our lower bound on $\beta$ approaches 1, corresponding to submodularity. We conjecture that this bound on $\beta$ is in fact the tightest achievable.


\section{Conclusion}
In this paper, we introduce the definition of upper submodularity ratio $\beta$ which complements the definition of \cite{Das:2011:Submodular}. We show that for two sparse subset selection problems: \emph{Sparse Regression} and \emph{Bayesian A-Optimal Design}, this ratio is bounded by spectral quantities. For functions with bounded upper and lower submodularity ratio, we give two algorithms with asymptotic performance $\frac{1}{2}$ and $1 - e^{-1}$, respectively. These algorithms yield state-of-the-art performance guarantees for the two applications we consider.
As open questions, \begin{itemize}
\item Can we characterize functions with bounded upper submodularity ratio. Elenberg et al. \cite{Elenberg:2018:Strong} showed that RSC implies weak submodularity. Does it imply bounded $\beta$?
\item It is still open to determine the approximation ratio of the greedy algorithm for weakly submodular function under a matroid constraint.
\item Is it possible to round the fractional solution of the multilinear relaxation of $(\gamma, \beta)$-weakly submodular function?
\end{itemize}

\bibliographystyle{plain}
\bibliography{minimal.bib}
\appendix

\section{Properties of $g_\phi$}
\label{sec:properties-g_phi}
Here we give further properties of the potential
\begin{equation*}
\label{eq:g-def}
g_\phi(A) =  \int_0^1\!\! \frac{\phi e^{\phi p}}{e^{\phi} - 1}\sum_{B \subseteq A}p^{|B|-1}(1-p)^{|A|-|B|}f(B)\,dp =\sum_{B \subseteq A}\mc{\phi}{|A|-1}{|B|-1}
f(B)
\end{equation*}
defined in Section~\ref{sec:non-oblivious-local}. We recall that the coefficients $\mc{\phi}{a}{b}$ for $0 \leq b \leq a$ are defined by
\begin{equation*}
\mc{\phi}{a}{b} = \int_{0}^1 \frac{\phi e^{\phi p}}{e^{\phi}-1}p^b(1-p)^{a-b} \,dp.
\end{equation*}
If we consider a continuous distribution $\cD_{\phi}$ on $[0,1]$ with density function:
\[\cD_{\phi}(x) = \frac{\phi e^{\phi x}}{e^\phi - 1}\]
then we can succinctly express these coefficients as $\mc{\phi}{a}{b} = \expect_{p \sim \cD_\phi}[p^a(1-p)^b]$. For convenience, we will define $\mc{\phi}{a}{b} = 0$ if either $a < 0$ or $b < 0$.

Here, and in the following sections, we define the function $h : \reals \to \reals$ by $h(x) = \frac{x e^x}{e^{x} - 1}$. Then, by Bernoulli's inequality $\frac{d h}{d x} = \frac{e^x(e^x - 1 - x)}{(e^x -1)^2} \geq 0$ and so $h$ is an increasing function. Note that our algorithm's claimed guarantee can then be expressed as $\frac{\gamma^2}{h(\phi(\gamma,\beta))}$ where $\phi(\gamma,\beta) = \gamma^2 + \beta(1 - \gamma)$. Finally, for $k \in \mathbb{Z}_+$, we let $H_k$ denote the $k^{\textrm{th}}$ harmonic number $H_k = \sum_{i = 1}^k 1/i = \Theta(\log k)$.

\subsection{Properties of the coefficients $\mc{\phi}{a}{b}$}
\label{sec:prop-coeff-mcph}
Here we provide a proof for the following properties $\mc{\phi}{a}{b}$ used in Section~\ref{sec:non-oblivious-local}, which we restate here for convenience. We give a proof of each claim in turn.

\coeffprops*

\begin{proof}[Proof of Claim 1]
Note that by the definition of $g_\phi$:
\begin{align*}
g_\phi(e | A) &=
\sum_{B \subseteq A + e} \mc{\phi}{|A|}{|B| - 1} f(B) - \sum_{B \subseteq A} \mc{\phi}{|A| - 1}{ |B| - 1}f(B) \\
&= \sum_{B \subseteq A} \left[(\mc{\phi}{|A|}{|B| - 1} - \mc{\phi}{|A|-1}{|B|-1})f(B) + \mc{\phi}{|A|}{|B|}f(B+e)\right].
\end{align*}
It thus suffices to show $(\mc{\phi}{|A|}{|B|-1} - \mc{\phi}{|A|-1}{|B|-1})f(B) = -\mc{\phi}{|A|}{|B|}f(B)$. For $B = \emptyset$, we have $f(\emptyset)  = 0$ and so $(\mc{\phi}{|A|}{-1} - \mc{\phi}{|A|-1}{-1})f(\emptyset) = 0 = -\mc{\phi}{|A|}{0}f(\emptyset)$. When $|B| \geq 1$,
\begin{multline*}
\mc{\phi}{|A|}{|B|-1} - \mc{\phi}{|A|-1}{|B|-1} = \expect_{p \sim \cD_{\phi}}\left[p^{|B|-1}(1-p)^{|A|-|B|+1} - p^{|B|-1}(1-p)^{|A|-|B|}\right] \\
= \expect_{p \sim \cD_\phi}\left[-p^{|B|}(1-p)^{|A|-|B|}\right] = -\mc{\phi}{|A|}{|B|}.
\end{multline*}
\end{proof}

\begin{proof}[Proof of Claim 2] By linearity of expectation:
\begin{equation*}
\sum_{B \subseteq A}\mc{\phi}{|A|}{|B|} = \sum_{b = 0}^{|A|} \binom{|A|}{b}\expect_{p \sim \cD_\phi}\ld p^b(1-p)^{|A|-b}\rd =
\expect_{p \sim \cD_\phi}\ld \sum_{b = 0}^{|A|} \binom{|A|}{b} p^b(1-p)^{|A|-b}\rd =
 1.
\end{equation*}
\end{proof}

\begin{proof}[Proof of Claim 3] When $0 \leq b \leq a$, the definition of $\mc{\phi}{a}{b}$ immediately gives:
\begin{equation*}
\mc{\phi}{a}{b} =\!\! \expect_{p \sim \cD_\phi}p^b(1-p)^{a-b}
=\!\! \expect_{p \sim \cD_\phi}[p^{b}(1-p)^{a-b}p + p^b(1-p)^{a-b}(1-p)]
= \mc{\phi}{a+1}{b+1} + \mc{\phi}{a+1}{b}.
\end{equation*}
\end{proof}

\begin{proof}[Proof of Claim 4] For $a > 0$ and $b \leq a$, noting that $\cD_\phi(p) = \frac{d}{dp} \frac{\cD_\phi(p)}{\phi}$ and applying integration by parts
\begin{multline*}
\mc{\phi}{a}{b} = \int_{0}^1 \cD_\phi(p) \cdot p^b (1-p)^{a-b}\,dp \\
= \left. \frac{\cD_\phi(p)}{\phi}p^b(1-p)^{a-b}\right|_{p = 0}^{p=1} - \int_{0}^1\frac{\cD_\phi(p)}{\phi} \left(b p^{b-1}(1-p)^{a-b} - (a-b) p^b (1-p)^{a-b-1}\right)\,dp\,.
\end{multline*}
Which is equivalent to:
\begin{equation*}
\phi \mc{\phi}{a}{b} = -b \mc{\phi}{a-1}{b-1} + (a-b) \mc{\phi}{a-1}{b} + \cD_\phi(p) p^b(1-p)^{a-b} \Bigr|_{p = 0}^{p=1}\,.
\end{equation*}
This follows immediately from the definition of $\mc{\phi}{a}{b}$ when $b > 0$, and when $b = 0$ it follows from $-b\mc{\phi}{a-1}{b-1} = 0 = bp^{b-1}(1-p)^{a-b}$. 

To complete the claim, we note that $\lim_{p \to 0^+} \cD_\phi(p) p^b(1-p)^{a-b}$ is $\cD_\phi(0) = \phi/(e^{\phi} - 1)$ if $b = 0$ and 0 if $b > 0$, and $\lim_{p \to 1^-} \cD_\phi(p) p^b(1-p)^{a-b}$ is $\cD_\phi(1) = \phi e^{\phi}/(e^{\phi} - 1)$ if $a = b$, and 0 if $0 \leq b < a$.
\end{proof}

\subsection{Bounding the value of $g$}
\label{sec:bounding-value-g}
Here we show that the value of $g_\phi(A)$ can be bounded in terms of $f(A)$ for any set $A$. In the analysis of~\cite{Filmus:2014}, this follows from submodularity of $g$, which is inherited from the submodularity of $f$. Here, we must again adopt a different approach. We begin by proving the following claim. Fix some set $A \subseteq X$ and for all $0 \leq j \leq |A|$ define
$F_j = \sum_{B \in \binom{A}{j}} f(B)$ as the total value of all subsets of $A$ of size $j$. Note that since we suppose $f$ is normalized, $F_0 = f(\emptyset) = 0$.
\begin{lemma}
\label{lem:lower-bound-helper}
If $f$ is $\gamma$-weakly submodular from below, then $F_i \geq \binom{|A|-1}{i-1} \gamma f(A)$ for all $1 \leq i \leq |A|$.
\end{lemma}
\begin{proof}[Proof of Lemma \ref{lem:lower-bound-helper}]
Let $k = |A|$. Since $f$ is $\gamma$-weakly submodular, for any $B \subseteq A$ we have \[\sum_{e \in A \setminus B}(f(B + e) - f(B)) \geq \gamma(f(A) - f(B)).\] Rearranging this, we have
\begin{equation}
\label{eq:g-lower-1}
\sum_{e \in A \setminus B}f(B+e) \geq \gamma f(A) + (|A| - |B| - \gamma)f(B) \geq
\gamma f(A) + (|A| - |B| - 1)f(B)\,,
\end{equation}
for all $B \subseteq A$. Summing~\eqref{eq:g-lower-1} over all $\binom{k}{j}$ possible subsets $B$ of size $j$, we obtain
\begin{equation}
  \label{eq:g-lower-2}
(j+1)F_{j+1} \geq \gamma \tbinom{k}{j} f(A) + (k - j - 1)F_j,
\end{equation}
since each set $T$ of size $j+1$ appears once as $B + e$ on the left-hand side of~\eqref{eq:g-lower-1} for each of the $j+1$ distinct choices of $e \in T$ with $B = T - e$.

We now show that $F_i \geq \binom{k-1}{i-1} \gamma f(A)$ for all $1 \leq i \leq k$. The proof is by induction on $i$.
For $i = 1$, the claim follows immediately from \eqref{eq:g-lower-2} with $j = 0$, since then $\binom{k}{j} = 1 = \binom{k-1}{i-1}$ and $(k - j - 1)F_j = (k-1)F_0 = 0$. For the induction step, \eqref{eq:g-lower-2} and the induction hypothesis imply:
\begin{align*}
F_{i+1} &\geq \textstyle
\frac{1}{i+1}\left(\gamma \binom{k}{i} f(A) + (k - i-1)F_i\right)
\geq \frac{1}{i+1}\left(\gamma \binom{k}{i} f(A) + (k - i-1)\gamma \binom{k - 1}{i- 1} f(A)\right) \\
&= \textstyle \frac{\gamma}{i+1}\left(\frac{k}{i}\binom{k-1}{i-1} f(A) + (k - i-1) \binom{k - 1}{i- 1} f(A)\right)
= \textstyle \frac{\gamma}{i+1}\left(\frac{k+ ki - i^2-i}{i}\right)\binom{k-1}{i-1} f(A) \\
&= \textstyle \frac{\gamma}{i+1} \frac{k(i+1) - i(i + 1)}{i}\binom{k-1}{i-1} f(A)
= \gamma \frac{k - i}{i}\binom{k-1}{i-1} f(A)
= \gamma \binom{k-1}{i} f(A).
\end{align*}
\end{proof}

Using the above claim, we now bound the value of $g_\phi(A)$ for any set $A$.
\begin{lemma}
  \label{lem:g-bounds}
  If $f$ is $\gamma$-weakly submodular, then for all $A \subseteq X$, $\gamma f(A) \leq g_\phi(A) \leq  h(\phi)H_{|A|} f(A)$.
\end{lemma}
\begin{proof}[Proof of Lemma \ref{lem:g-bounds}]
Let $k = |A|$. We begin with the lower bound for $g_\phi(A)$. By the definition of the coefficients $\mc{\phi}{a}{b}$ and Lemma~\ref{lem:lower-bound-helper}:
\begin{align*}
g_\phi(A) &= \sum_{B \subseteq A}\expect_{p \sim \cD_\phi}\ld p^{|B|-1}(1-p)^{|A| - |B|}\rd f(B)
= \expect_{p \sim \cD_\phi}\biggl[ \sum_{i = 1}^k p^{i-1}(1-p)^{k - i}F_i\biggl] \\
&\geq \expect_{p \sim \cD_\phi}\biggl[ \sum_{i = 1}^k \gamma p^{i-1}(1-p)^{k - i}\tbinom{k-1}{i-1}f(A) \biggr]
= \expect_{p \sim \cD_\phi}\biggl[\gamma f(A) \sum_{i = 0}^{k-1} \tbinom{k-1}{i}p^{i}(1-p)^{k - i - 1}\biggr] \\
&= \expect_{p \sim \cD_\phi}[ \gamma f(A)] = \gamma f(A).
\end{align*}

For the upper bound, we similarly have:
\begin{align*}
g_\phi(A) &= \expect_{p \sim \cD_\phi}\biggl[ \sum_{i = 1}^k p^{i-1}(1-p)^{k - i}F_i\biggl] \leq \expect_{p \sim \cD_\phi}\biggl[ \sum_{i = 1}^k p^{i-1}(1-p)^{k - i}\binom{k}{i}f(A)\biggl]\\\
&= \int_{0}^1 \frac{\phi e^{\phi p}}{e^\phi - 1}\frac{\sum_{i=1}^{k} \binom{k}{i} p^{i}(1-p)^{k - i}f(A)}{p} \,dp
= \int_{0}^1 \frac{\phi e^{\phi p}}{e^\phi - 1}\frac{1 - (1-p)^{k}}{p} f(A) \,dp
       \\
&\leq \frac{\phi e^{\phi}}{e^\phi - 1}\int_{0}^1 \frac{1-(1-p)^{k}}{p}f(A)\,dp
       = h(\phi)\int_{0}^1 \sum_{j = 0}^{k-1}(1-p)^jf(A)\,dp \\
       &= h(\phi)\sum_{j = 0}^{k-1}\frac{1}{j+1}f(A) = h(\phi)H_{|A|}f(A),
\end{align*}
where the first inequality follows from monotonicity of $f$ and the second inequality from $\frac{1 - (1-p)^{k}}{p} > 0$ for $p \in (0,1]$ and $h(\phi) = \frac{\phi e^\phi}{e^{\phi}-1}$ is an increasing function of $p$.
\end{proof}

\subsection{Bounding the sensitivity of $g_\phi$ to $\phi$}
\label{sec:bound-sens-g_phi}
The following lemma shows that small changes in the parameter $\phi$ produce relatively small changes in the value $g_\phi(A)$ for any set $A$.
\begin{lemma}
\label{lem:g-change-ineq}
For all $\phi$, $\epsilon > 0$, and $S \subseteq X$,
\begin{enumerate}
\item $g_{\phi(1-\epsilon)}(S) \geq e^{-\phi\epsilon}g_\phi(S)$
\item $h(\phi) \leq e^{\phi\epsilon}h(\phi(1-\epsilon))$
\end{enumerate}
\end{lemma}
\begin{proof}[Proof of Lemma \ref{lem:g-change-ineq}]
Both claims will follow from the inequality
\begin{equation}
\label{eq:g-change-1}
\frac{\phi(1-\epsilon)e^{\phi(1-\epsilon)p}}{e^{\phi(1-\epsilon)} - 1} \geq
e^{-\phi\epsilon}\frac{\phi e^{\phi p}}{e^{\phi}-1}\,,
\end{equation}
which we show is valid for all $p \in [0,1]$ and $\epsilon > 0$. Indeed, under these assumptions,  \begin{multline*}
\frac{\phi(1-\epsilon)e^{\phi(1-\epsilon)p}}{e^{\phi(1-\epsilon)} - 1}\cdot \frac{e^\phi - 1}{\phi e^{\phi p}}
= (1-\epsilon)e^{-\phi \epsilon p}\frac{e^{\phi} - 1}{e^{\phi}e^{-\phi\epsilon} - 1}
= (1-\epsilon)e^{-\phi \epsilon p}\frac{e^{\phi} - 1}{e^{\phi}(1 + (e^{-\phi}-1))^\epsilon - 1}
\\
\geq (1-\epsilon)e^{-\phi \epsilon p}\frac{e^{\phi} - 1}{e^{\phi}(1 +  \epsilon(e^{-\phi}- 1)) - 1}
= (1-\epsilon)e^{-\phi \epsilon p}\frac{e^{\phi} - 1}{(1-\epsilon)(e^{\phi}-1)}  = e^{-\phi\epsilon p} \geq e^{-\phi \epsilon}\,.
\end{multline*}
Here the first inequality follows from the generalized Bernoulli inequality $(1 + x)^t \leq (1+ t x)$, which holds for all $x \geq -1$ and $0 \leq t \leq 1$, and the second inequality follows from $p \in [0,1]$.

For the first claim, applying \eqref{eq:g-change-1} gives
\begin{multline*}
g_{\phi(1-\epsilon)}(A) = \int_{0}^{1}\frac{\phi(1-\epsilon) e^{\phi(1-\epsilon)p}}{e^{\phi(1-\epsilon)p}-1}\sum_{B \subseteq A}p^{|B|-1}(1-p)^{|A|-|B|}f(B)\,dp \\
\geq \int_{0}^{1}e^{-\phi \epsilon} \frac{\phi e^{\phi p}}{e^{\phi}-1}\sum_{B \subseteq A}p^{|B|-1}(1-p)^{|A|-|B|}f(B)\,dp
= e^{-\phi\epsilon}g_\phi(A),
\end{multline*}
as required. For the second claim, setting $p = 1$ in \eqref{eq:g-change-1} gives $h(\phi(1-\epsilon)) \geq e^{-\phi\epsilon}h(\phi)$ or, equivalently, $h(\phi) \leq e^{\phi\epsilon}h(\phi(1-\epsilon))$.
\end{proof}

\subsection{Efficiently estimating $g_\phi$ via sampling}
\label{sec:effic-estim-g_phi}
The definition of $g_\phi$ requires evaluating $f(B)$ on all $B \subseteq A$, which requires $2^{|A|}$ calls to the value oracle for $f$. In this section, we show that we can efficiently estimate $g_\phi$ using only a polynomial number of value queries to $f$. Our sampling procedure is based on the same general ideas described in~\cite{Filmus:2014}, but here we focus on evaluating only the \emph{marginals} of $g_\phi$, which results in a considerably simpler implementation. In particular, our algorithm does not require computation of the coefficients $\mc{\phi}{a}{b}$.
\begin{lemma}
\label{lem:g-sampling}
For any $\phi$, $N$, there is a randomized procedure for obtaining an estimate $\tilde{g}(e|A)$ of $g_\phi(e | A)$ using $N$ queries to the value oracle for $f$ so that for any $\delta > 0$,
\[
\prob{\,|g(e | A) - \tilde{g}(e|A)| \geq \delta f(A+e)\,} <
2e^{-\frac{\delta^2 N}{2}}\,,
\]
\end{lemma}
\begin{proof}[Proof of Lemma \ref{lem:g-sampling}]
We consider the following 2-step procedure given as an interpretation of $g$  in~\cite{Filmus:2014}: we first sample $p \sim \cD_{\phi}$, then construct a random $B \subseteq A$ by taking each element of $A$ independently with probability $p$. The probability that any given $B \subseteq A$ is selected by the procedure is then precisely
\[
\int_{0}^1\frac{\phi e^{\phi p}}{e^\phi - 1}p^{|B|}(1-p)^{|A|}\,dp = \mc{\phi}{|A|}{|B|}\,.
\]
Thus, for a random $\tilde{B} \subseteq A$ sampled in this fashion, $\expect[f(e|\tilde{B})] = \sum_{B \subseteq A}\mc{\phi}{|A|}{|B|}f(e | B) = g(e |A)$, by part 1 of Lemma~\ref{lem:coeff-props}. We remark that for the particular distributions $\cD_\phi$ we consider, the first step of the procedure can easily by implemented with inverse transform sampling.

Suppose now that we draw $N$ independent random samples $\{B_i\}_{i=1}^N$ in this fashion and define the random variables $Y_i = \frac{g(e|A) - f(e|B_i)}{f(A + e)}$. Then, $\expect[Y_i] = 0$ for all $i$. Moreover, by monotonicity of $f$, $0 \leq f(e|B) \leq f(B+e) \leq f(A+e)$ for all $B \subseteq A$ and also $0 \leq \sum_{B \subseteq A}\mc{\phi}{|A|}{|B|}f(e | B) = g(e |A)$ and $g(e | A) = \sum_{B \subseteq A}\mc{\phi}{|A|}{|B|}f(e | B) \leq \sum_{B \subseteq A}\mc{\phi}{|A|}{|B|}f(A + e) = f(A + e)$ by part 2 of Lemma~\ref{lem:coeff-props}. Thus, $|Y_i| \leq 1$ for all $i$. Let $\tilde{g}_\phi(e|A) = \frac{1}{N}\sum_{i = 1}^N f(e|B_i)$. Applying the Chernoff bound (Lemma~\ref{lem:chernoff-bound}), for any $\delta > 0$ we have
\begin{equation*}
\prob{\,|g(e | A) - \tilde{g}(e|A)| \geq \delta f(A+e)\,} \leq
\prob{\textstyle \sum_{i = 1}^N Y_i > \delta N} < 2e^{-\frac{\delta^2 N}{2}}. \qedhere
\end{equation*}
\end{proof}

\section{A randomized, polynomial time distorted local-search algorithm}
\label{sec:our-final-algorithm}
Our final algorithm is shown in Algorithm~\ref{alg:distorted-ls-full}. Before presenting it in detail, we describe the main concerns involved in its formulation.

\subsection{Initialization}
\label{sec:initialization}
We initialize the algorithm with a solution $S_0$ by using the guarantee for \textsc{ResidualRandomGreedy} provided by~\cite{Chen:2018:Weakly} when only $\gamma$ is bounded. In this case, their analysis shows that the \emph{expected} value of the solution produced by the algorithm is at least $\frac{1}{(1+\gamma^{-1})^2}f(O)$, where $O$ is an optimal solution to the problem. Here, however, we will require a guarantee that holds with high probability. This is easily ensured by independently running \textsc{ResidualRandomGreedy} a sufficient number of times and taking the best solution found.

Formally, suppose we set $\e' = \min(\e,\frac{1}{128})$ and run \textsc{ResidualRandomGreedy} $G = \frac{2\log(n)}{\epsilon'^2} = \tilde{\cO}(\epsilon^{-2})$ times independently. For each $1 \leq l \leq G$, let $T_l$ be the solution produced by the $l^\textrm{th}$ instance of the \textsc{ResidualRandomGreedy}. Define the random variables $Z_l = \frac{1}{(1+\gamma^{-1})^2} - \frac{f(T_l)}{f(O)}$, where $O \in \cI$ is the optimal solution. Then, $\expect[Z_l] = 0$ and $|Z_l| \leq 1$ for all $l$. 

Let $S_0 = \argmax_{1 \leq \l \leq G} f(T_l)$. Then, by the Chernoff bound (Lemma~\ref{lem:chernoff-bound}),
\begin{multline*}
\textstyle
\prob{f(S_0) < \left((1+\gamma^{-1})^{-2} - \epsilon'\right)f(O)} \leq
\prob{\frac{1}{G}\sum_{l = 1}^G f(T_l) < \left((1+\gamma^{-1})^{-2} - \epsilon'\right)f(O)} \\ \textstyle
= \prob{\sum_{l = 1}^G Z_l > G\epsilon'}  < e^{-\frac{\epsilon'^2 G}{2}} = \frac{1}{n}.
\end{multline*}
Thus, with probability at least $1 - \frac{1}{n} = 1 - o(1)$, $f(S_0) \geq \left(\frac{1}{(1+\gamma^{-1})^2} - \epsilon'\right)f(O)$.

\subsection{Determining $\phi$}
\label{sec:determining-phi}
In Theorem~\ref{thm:loc-opt-main}, we considered a $(\gamma,\beta)$-weakly submodular function $f$, and used the potential $g_\phi$ with $\phi = \phi(\gamma,\beta) = \gamma^2 + \beta(1-\gamma)$ to guide the search. In general, however, the values of $\gamma$ and $\beta$ may not be known in advance. One approach to coping with this would be to make an appropriate series of guesses for each of the values, then run our the algorithm for each guess and return the best solution obtained.

Here we describe an alternative and more efficient approach: we guess the value of $\phi(\gamma,\beta)$ directly from an appropriate range of values. Moreover, when running the algorithm for each subsequent guess, we initialize the local search procedure using the solution produced by the algorithm for the previous guess. Combined with the bounds from \ref{lem:g-change-ineq}, this will allow us to amortize the number of improvements made by the algorithm across all guesses.

In the next lemma, we show that if $\gamma$ or $\phi(\gamma,\beta)$ is very small, then guarantee for \textsc{ResidualRandomGreedy} is stronger than that required by our analysis (and so $S_0$ is already a good solution).
This will allow us to bound the range of values for both $\phi$ and $\gamma$ that we must consider in our algorithm.
\begin{lemma}
For all $\gamma \in (0,1]$ and $\beta \geq 1$, $\phi(\gamma,\beta) \geq \frac{3}{4}$. Moreover, if $\phi(\gamma,\beta) > 4$ or $\gamma < \frac{1}{7}$, then $\frac{1}{(1+\gamma^{-1})^2} > \frac{\gamma^2(1 - e^{-\phi(\gamma,\beta)})}{\phi(\gamma,\beta)}$.
\label{lem:phi-gamma-range}
\end{lemma}
\begin{proof}[Proof of Lemma \ref{lem:phi-gamma-range}]
First, we show that $\phi(\gamma,\beta) \geq 3/4$ for any value of $\gamma \in (0,1]$ and $\beta \geq 1$. Note that $\frac{\partial \phi}{\partial \beta} = 1 - \gamma \geq 0$, for all $\gamma \in [0,1]$. Thus, any minimizer of $\phi(\gamma,\beta)$ sets $\beta = 1$. Moreover, $\frac{\partial \phi}{\partial \gamma} = 2\gamma - \beta$ and $\frac{\partial^2 \phi}{\partial \gamma^2} = 2$ so a $\phi(\gamma,\beta)$ is minimized by $\gamma = \frac{\beta}{2} = \frac{1}{2}$. It follows that $\phi(\gamma,\beta) \geq \phi\left(\frac 1 2 , 1\right) = \frac{3}{4}$ for all $\gamma \in [0,1]$ and $\beta \geq 1$.

Now suppose that $\phi(\gamma,\beta) > 4$. Then, the claim follows, since
\begin{equation*}
\frac{\gamma^2(1-e^{-\phi(\gamma,\beta)})}{\phi(\gamma,\beta)}
< \frac{\gamma^2}{4} \leq \frac{\gamma^2}{(1+\gamma)^2} =
\frac{1}{(1 + \gamma^{-1})^2}\,.
\end{equation*}
It remains to consider the case in which $\gamma < \frac{1}{7}$. Recall that $h(x)\triangleq\frac{xe^x}{e^x-1}$ is increasing in $x$ and so $h(\phi(\gamma,\beta)) \geq h(\frac{3}{4}) > \frac{4}{3}$ (where the last inequality follows directly by computation of $h(\frac{3}{4})$). Suppose that $\gamma < \frac{1}{7}$. Then,
\begin{equation*}
\frac{\gamma^2(1 - e^{-\phi(\gamma,\beta)})}{\phi(\gamma,\beta)}
= \frac{\gamma^2}{h(\phi(\gamma,\beta))} <
\tfrac{3}{4}\gamma^2\,.
\end{equation*}
Comparing the previous estimation to the approximation ratio of \cite{Chen:2018:Weakly} and using that $\gamma < 1/7$, we have
\begin{equation*}
\frac{\frac{3}{4}\gamma^2}{(1+\gamma^{-1})^{-2}} = \frac{3}{4}(\gamma + 1)^2 < \frac{3}{4}\left(\frac{8}{7}\right)^2 < 1.
\end{equation*}
Thus, $\frac{3}{4}\gamma^2 < \frac{1}{(1+\gamma^{-1})^2}$ and again the claim follows.
\end{proof}

Lemma~\ref{lem:phi-gamma-range} shows that it suffices to consider $\phi(\gamma, \beta) \in [3/4, 4]$ and $\gamma > 1/7$, since otherwise the starting solution already satisfies the claimed guarantee.\footnote{We remark that the use of the \textsc{ResidualRandomGreedy} is not strictly necessary for our results. One can instead initialize the algorithm with a base containing the best singleton as in the standard local search procedure to obtain a guarantee of $\gamma/k$ for the initial solution. The remaining arguments can then be modified at the cost of a larger running time dependence on the parameter $k$.} Thus, our algorithm considers a geometrically decreasing sequence of guesses for the value $\phi \in [3/4,4]$, given by $\phi_j = 4(1-\e)^j$, where $0 \leq j \leq \lceil\log_{1-\e}\frac{3}{16}\rceil$. For the first guess, we initialize our algorithm with the solution $S_0$ produced using several runs of \textsc{ResidualRandomGreedy}. For each guess after the this, we initialize $S$ with the approximately locally optimal solution produced for the previous guess. 

For each guess, the algorithm proceeds by repeatedly searching for single element swaps that significantly improve the potential $g_\phi(S)$. Specifically, we will exchange an element $a \not\in S$ with an element $b \in S$ whenever $\tilde{g}_{\phi_j}(a | S - b) > \tilde{g}_{\phi_j}(b|S - b) + \Delta f(S)$, where $\tilde{g}_{\phi_j}( \cdot | S-b)$ is an estimate of $g_{\phi_j}(\cdot | S-b)$ computed using $N$ samples as described in Section~\ref{sec:effic-estim-g_phi} and $\Delta$ is an appropriately chosen parameter. We show that by setting $N$ appropriately, we can ensure that with high probability an approximate local optimum of every $g_{\phi}$ is reached after at most some total number $M$ of improvements across all guesses.

\subsection{The algorithm and its analysis}
\label{sec:final-algorithm}
Our final algorithm is shown in Algorithm~\ref{alg:distorted-ls-full}. Let $\cM = (X,\cI)$ be a matroid, and $f : 2^X \to \posreals$ be a $(\gamma,\beta)$-weakly submodular function.  Given some $0 < \epsilon \leq 1$ we set the parameters:
\begin{align}
\Delta &= \tfrac{\epsilon}{k} \tag{threshold for accepting improvements} \\
\delta &= \tfrac{\Delta}{4h(4) \cdot H_k} = \tfrac{\epsilon}{4h(4)\cdot H_kk} = \Theta(\epsilon/(k\log k)) \tag{bound on sampling accuracy} \\
L &= 1+\lceil \log_{1-\e} \tfrac{3}{16} \rceil = \cO(\e^{-1}) \tag{number of guesses for $\phi$} \\
M &= \log_{1 + \delta}(7\cdot 128 \cdot e^{4L\epsilon}h(4)\cdot H_k) = \tilde{\cO}(\delta^{-1}) = \tilde{\cO}(k\e^{-1}) \tag{total number of improvements} \\
N &= 4\cdot 7^2 \delta^{-2}\ln(Mkn) = \tilde{\cO}(\delta^{-2}) = \tilde{\cO}(k^2\e^{-2}) \tag{number of samples to estimate $g_\phi$}
\end{align}
\begin{algorithm2e}[t]
\SetKw{Break}{break}
\SetKw{True}{true}
\SetKw{False}{false}
Let $\Delta = \frac{\epsilon}{k}$, $\delta = \frac{\Delta}{4h(4)\cdot H_k} = \frac{\epsilon}{4h(4)\cdot H_kk}$, $M = (1+\delta^{-1})(37 + \ln(H_k))$, $N = 28\delta^{-2}\ln(Mkn)$, $G=\log(n)/(2\min(\e,\frac{1}{128})^2)$\;
$S_0 \gets$ the best output produced by $G$ independent runs of \textsc{ResidualRandomGreedy} applied to $f$ and $\cM$\;
$S_{\max} \gets S_0$\;
$i \gets 0$\;
\For{$0 \leq j \leq \lceil\log_{1-\epsilon}16/3\rceil$}{
  $\phi \gets 4(1-\epsilon)^j$\;
  $S \gets S_j$\;
  \Repeat{$\mathrm{isLocalOpt}$ or $i \geq M$}{
    $\mathrm{isLocalOpt} \gets \True$\;
    \ForEach{$b \in S$ and $a \in X \setminus S$ with $S - b + a \in \cI$}{
        Compute $\tilde{g}_{\phi_j}(a | S - b)$ and $\tilde{g}_{\phi_j}(b | S-b)$ using $N$ random samples\;
        \If{$\tilde{g}_{\phi_j}(a | S - b) > \tilde{g}_{\phi_j}(b | S - b) + \Delta f(S)$}
        {
          $S \gets S - b + a$\;
          $i \gets i+1$\;
          $\mathrm{isLocalOpt} \gets \False$\;
          \Break
        }
      }
  }
  $S_{j+1} \gets S$\;
  \lIf{$f(S_{j+1}) > f(S_{\max})$}{$S_{\max} \gets S_{j+1}$}
}
\Return{$S_{\mathrm{max}}$}
\caption{Distorted Local Search Implementation}
\label{alg:distorted-ls-full}
\end{algorithm2e}
In Algorithm~\ref{alg:distorted-ls-full}, we evaluate potential improvements using an estimate $\tilde{g}_{\phi_j}(\cdot | S-b)$ for the marginals of $g$ that is computed using $N$ samples. By Lemma~\ref{lem:g-sampling}, we then have $|\tilde{g}_{\phi_j}(e | A) - g_{\phi_j}(e | A)| \leq \gamma\delta f(A+e)$ for any $A,e$ considered by the algorithm with probability at least $1 - 2e^{-\frac{\delta^2 \gamma^{2} N}{2}}$.
If $\gamma \geq 1/7$, this is at least $1 - 2e^{-\frac{\delta^2}{2\cdot 7^2} N} = 1 - \frac{2}{(Mkn)^2}$. In our algorithm we will limit the total number of improvements made across all guesses for $\phi$ to be at most $M$. Note that any improvement can be found by testing at most $kn$ marginal values, so we must estimate at most $Mkn$ marginal values across the algorithm. By a union bound, we then have $|\tilde{g}_{\phi_j}(e | A) - g_{\phi_j}(e | A)| \leq \gamma\delta f(A+e)$ for \emph{all} $A,e$ considered by Algorithm~\ref{alg:distorted-ls-full} with probability at least $1 - o(1)$ whenever $\gamma \geq 1/7$.
Before proving our main result, let us show that if the algorithm terminates and returns $S$ after making $M$ improvements, we must in fact have an \emph{optimal} solution with high probability.

\begin{lemma}
\label{lem:max-improvements}
Suppose that $\gamma \geq 1/7$. Then, if Algorithm~\ref{alg:distorted-ls-full} makes $M$ improvements, the set $S$ it returns satisfies $f(S) \geq f(O)$ with probability $1 - o(1)$.
\end{lemma}
\begin{proof}[Proof of Lemma \ref{lem:max-improvements}]
With probability $1 - o(1)$ we have $\left|\tilde{g}_{\phi_j}(e | A) - g_{\phi_j}(e | A)\right| \leq \gamma\delta f(A+e)$ for any $e,A$ considered by Algorithm~\ref{alg:distorted-ls-full}. Whenever the algorithm exchanges some $a \in X \setminus S$ for $b \in S$ for some guess $\phi_j$, we have $\tilde{g}_{\phi_j}(a | S-b) - \tilde{g}_{\phi_j}(b | S-b) \geq \Delta f(S)$ and so
\begin{align*}
g_{\phi_j}(S - b + a) - g_{\phi_j}(S) &=
g_{\phi_j}(a | S - b) - g_{\phi_j}(b | S - b) \\
&\geq \tilde{g}_{\phi_j}(a | S - b) - \delta\gamma f(S - b + a) - \tilde{g}_{\phi_j}(b | S-b) - \delta\gamma f(S) \\
&\geq \tilde{g}_{\phi_j}(a | S - b) - \delta g_{\phi_j}(S - b + a) - \tilde{g}_{\phi_j}(b | S-b) - \delta g_{\phi_j}(S) \\
&\geq \Delta f(S) - \delta g_{\phi_j}(S - b + a) - \delta g_{\phi_j}(S),
\end{align*}
where the second inequality follows from the lower bound on $g_{\phi_j}$ in Lemma~\eqref{lem:g-bounds}. Rearranging and using the upper bound on $g_{\phi_j}(S)$ from Lemma~\ref{lem:g-bounds}, together with the definition of $\delta$ and $\Delta$, we obtain:
\begin{multline}
g_{\phi_j}(S - b + a) \geq \frac{\Delta f(S) + (1 - \delta)g_{\phi_j}(S)}{1 + \delta}
\geq \frac{\frac{\epsilon}{k}\frac{1}{h({\phi_j})\cdot H_k} + 1 - \delta}{1+\delta} g_{\phi_j}(S)  \\
\geq \frac{\frac{\epsilon}{k}\frac{1}{h(4)\cdot H_k} + 1 - \delta}{1+\delta} g_{\phi_j}(S)
= \frac{1 + 3\delta}{1+\delta} g_{\phi_j}(S) \geq (1 + \delta)g_{\phi_j}(S),
\label{eq:approx-improvement}
\end{multline}
where the last inequality follows from $\frac{1+3x}{1+x} \geq \frac{(1 + x)^2}{1+x}$ for all $0 \leq x \leq 1$.

Now suppose that $f(S_0) \geq ((1+\gamma^{-1})^{-2}-\e')f(O)$, which we have shown also occurs with high probability $1-o(1)$. Then, since $\gamma \geq \frac{1}{7}$ and $\e' = \min(\frac{1}{128},\e)$, we have $f(S_0) \geq \frac{1}{128}f(O)$. Suppose that $i = M$ when the algorithm is considering some guess $\phi_l$. We consider how the current value of $g_{\phi_j}(S)$ changes throughout Algorithm~\ref{alg:distorted-ls-full}, both as improvements are made and as $j$ increases. As shown in~\eqref{eq:approx-improvement}, each of our $M$ improvements increases this value by a factor of $(1+\delta)$. Moreover, as shown in Lemma~\ref{lem:g-change-ineq},
\[
g_{\phi_j}(S) = g_{(1-\epsilon)\phi_{j-1}}(S) \geq e^{-\phi_{j-1}\epsilon}g_{\phi_{j-1}}(S) \geq e^{-4\epsilon}g_{\phi_{j-1}}(S)\,,
\]
for any set $S$. Thus, each time $j$ is incremented, the value $g_{\phi_j}(S)$ decreases by a factor of at most $e^{4\epsilon}$. Since we made $M$ improvements, we then have:
\begin{equation*}
g_{\phi_l}(S_{\ell+1}) \geq (1+\delta)^Me^{-4l\epsilon}g_{\phi_0}(S_0)
\geq (1+\delta)^Me^{-4l\epsilon}\gamma f(S_0)
\geq (1+\delta)^Me^{-4l\epsilon}\tfrac{1}{7}\tfrac{1}{128}f(O)\,,
\end{equation*}
where the second inequality follows from the lower bound on $g$ given in Lemma~\ref{lem:g-bounds}, and the second from $\gamma \geq \frac{1}{7}$.
The upper bound on $g$ given by Lemma~\ref{lem:g-bounds} implies that:
$g_{\phi_l}(S_{l+1}) \leq h(\phi_l)H_kf(S_{l+1}) \leq h(4)H_kf(S_{l+1})$. Thus,
\begin{equation*}
f(S_{l+1}) \geq (1+\delta)^Me^{-4l \epsilon}\frac{1}{7 \cdot 128 \cdot h(4) \cdot H_k}f(O)\,.
\end{equation*}
Since $l \leq L$ and $M = \log_{1+\delta}(e^{4L\epsilon}7\cdot 128 \cdot h(4) \cdot H_k)$, the set $S_{\mathrm{max}}$ returned by the algorithm thus has $f(S_{\mathrm{max}}) \geq f(S_{l + 1}) \geq f(O)$, as claimed.
\end{proof}

We are now ready to prove our main claim, from Section~\ref{sec:non-oblivious-local}, restated here for convenience:
\dlsmain*
\begin{proof}[Proof of Theorem \ref{thm:distorted-ls-main}]
We have shown that $f(S_0) \geq \left((1+\gamma^{-1})^{-2} - \e'\right)f(O)$ with probability $1 - o(1)$, where $\e' = \min(\e,\frac{1}{128})$. If $\gamma < 1/7$ or $\phi(\gamma,\beta) \not\in [3/4,4]$, then Lemma~\ref{lem:phi-gamma-range} implies that $(1+\gamma^{-1})^{-2} > \frac{\gamma^2 (1 - e^{-\phi(\gamma,\beta)})}{\phi(\gamma,\beta)}$, and so the claim follows as $f(S_\mathrm{max}) \geq f(S_0)$. Thus, we suppose that $\gamma \geq 1/7$ and $\phi(\gamma,\beta) \in [3/4,4]$. Then, if Algorithm~\ref{alg:distorted-ls-full} makes $M$ improvements, Lemma~\ref{lem:max-improvements} implies that the set returned by the algorithm is optimal with probability at least $1 - o(1)$.

In the remaining case, we have $\phi(\gamma,\beta) \in [3/4, 4]$, $\gamma \geq 1/7$, and each set $S_{j+1}$ produced by the algorithm must have
 $\tilde{g}_{\phi_j}(o_l | S_{j+1} - s_l) \leq \tilde{g}_{\tilde{\phi_j}}(s_l | S_{j+1} - s_l) + \Delta f(S)$ for every $s_l \in S$ and $o_l \in O$. Since $\gamma \geq 1/7$ and the algorithm makes at most $M$ improvements, with probability $1-o(1)$, we have $\left|\tilde{g}_{\phi_j}(e | A) - g_{\phi_j}(e | A)\right| \leq \gamma\delta f(A+e)$ for all guesses $\phi_j$ and $e,A$ considered by the algorithm. Thus,
\begin{align*}
g_{\phi_j}&(S_{j+1} - s_l + o_l) - g_{\phi_j}(S_{j+1}) \\
&= g_{\phi_j}(o_l | S_{j+1} - s_l) - g_{\phi_j}(s_l | S_{j+1} - s_l) \\
&\leq \tilde{g}_{\phi_j}(o_l | S_{j+1} - s_l) + \delta\gamma f(S_{j+1}-s_l+o_l) - \tilde{g}_{\phi_j}(s_l | S_{j+1} - s_l) + \delta\gamma f(S_{j+1}) \\
&\leq \Delta f(S_{j+1}) + \delta\gamma f(S_{j+1}) + \delta\gamma f(S_{j+1}-s_l + o_l) \\
&\leq (\Delta + 2\delta)f(O) \\
&= O\bigl(\tfrac{\epsilon}{k}\bigr)\cdot f(O).
\end{align*}
Consider the smallest $j$ such that $\phi_{j+1}  \triangleq 4(1-\e)^{j+1} < \phi(\gamma,\beta)$. Then, $\phi_{j+1} < \phi(\gamma, \beta) \leq \phi_{j+1}/(1-\e) \triangleq \phi_{j}$. Let $\tilde{\beta} = \frac{\phi_{j} - \gamma^2}{1-\gamma}$. Then, $\phi(\gamma,\tilde{\beta}) = \gamma^2 + \frac{\phi_{j
} - \gamma^2}{1-\gamma}(1-\gamma) = \phi_j$ and $\tilde{\beta} \geq \frac{\phi(\gamma, \beta) -\gamma^2}{1-\gamma} = \beta$, so $f$ is also $(\gamma,\tilde{\beta})$-weakly submodular. Theorem~\ref{thm:loc-opt-main} then implies
\begin{equation*}
f(S_{j+1}) \geq \frac{\gamma^2}{h(\phi_j)}f(O) + \sum_{i = 1}^k \left[g_{\phi_j}(S)  - g_{\phi_j}(S - s_l + o_l)\right] \geq \left(\frac{\gamma^2}{h(\phi_j)} - \cO(\epsilon)\right)f(O)
\end{equation*}
By Lemma~\ref{lem:g-change-ineq} part 2, our choice of $j$, and $\phi_j \leq 4$,
\begin{equation*}
h(\phi_j) \leq e^{\phi_j\e}h((1-\e)\phi_j) \leq
e^{\phi_j\e}h(\phi(\gamma,\beta)) \leq
e^{4\e}h(\phi(\gamma,\beta)) 
\end{equation*}
Thus, $f(S_{j+1}) \geq \left(\frac{\gamma^2}{h(\phi(\gamma,\beta))} - \cO(\e)\right)f(O) = \left(\frac{\gamma^2(1 - e^{-\phi(\gamma,\beta)})}{\phi(\gamma,\beta)} - \cO(\e)\right)\!f(O)$.

The running time of the algorithm is dominated by the number of value oracle queries made to $f$. The initialization requires running \textsc{ResidualRandomGreedy} $\tilde{\cO}(\e^{-2})$ times, each of which requires $\cO(nk)$ value queries. The remaining execution makes at most $M = \tilde{\cO}(\e^{-1}k)$ local search improvements, each requiring at most $Nnk = \tilde{\cO}(nk^3\e^{-2})$ value queries to find. Altogether the running time is thus at most $\tilde{\cO}(nk^4\e^{-3})$.
\end{proof}


%


\section{A-optimal design for Bayesian linear regression}
\label{sec:optimal-design}

In Bayesian linear regression, we suppose data is generated by a linear model $\vy = X^T \vtheta + \veps$, where $\vy \in \mathbb{R}^n, X \in \mathbb{R}^{p \times n}$ and $\veps \sim \mathcal{N}(0, \sigma^2 I)$, where $I$ is the identity matrix.
Here,
$X = \begin{bmatrix}
    \vx_1 & \vx_2 & \cdots & \vx_n
\end{bmatrix}$
with $\vx_i \in \mathbb{R}^p$ is a vector of data, and $\vy$ is a vector corresponding observations for the response variable. The variable $\ve$ represents Gaussian noise with $0$ mean and variance $\sigma^2$. When the number of columns $n$ (i.e., the number of potential observations) is very large, \emph{experimental design} focuses on selecting a small subset $S \subset \{1, 2, \ldots, n\}$ of columns of $X$ to maximally reduce the variance of the estimator $\vtheta$.

Let $X_S, \vy_S$ be the matrix $X$ (the vector $\vy$ respectively) restricted to  columns (rows respectively) indexed by $S$. From classical statistical theory, the optimal choice of parameters for any such $S$ is given by $\hat{\vtheta}_S = (X_S^T X_S)^{-1} X_S \vy_S$ and satisfies $\var(\hat{\vtheta}_S) = \sigma^2 (X_S^T X_S)^{-1}$. Because the variance of $\hat{\vtheta}_S$ is a matrix, there is not a universal function which one tries to minimize to find the appropriate set $S$. Instead, there are multiple objective functions depending on the context leading to different optimality criteria.

As in \cite{Krause:2008vo, Bian:2017:Guarantees,DBLP:conf/icml/HarshawFWK19}, we consider the \emph{A-optimal} design objective. We suppose our prior probability distribution has $\vtheta \sim \mathcal{N}(0, \Lambda)$. We start by stating a standard result from Bayesian linear regression.

\begin{lemma}
    Given the previous assumption, and the prior on $\vtheta \sim \mathcal{N}(0, {\Lambda})$, The posterior distribution of $\vtheta$ follows a normal distribution
$p(\vtheta | \vy_S) \sim \mathcal{N}(M_S^{-1}X_S \vy_S, \, M_S^{-1})$, where $M_S^{-1} = \lb \sigma^{-2} X_S X^T_S + \Lambda^{-1}\rb^{-1}$.
\end{lemma}
In A-optimal design, our objective function seeks to reduce the variance of the posterior distribution of $\vtheta$ by reducing the trace of $M_{S}^{-1}$, i.e., the sum of the variance of the regression coefficients. Mathematically, we seek to maximize the following objective function
\begin{align}
    F(S) = \tr(\Lambda) - \tr(M_S^{-1}) = \tr(\Lambda) - \tr((\Lambda^{-1} + \sigma^{-2} {X}_S {X}^T_S)^{-1}).
\end{align}
The function $F$ is not submodular as shown in \cite{Krause:2008vo}. The current tightest estimation of the lower weak-submodular ratio of $F$ is due to Harshaw et al. \cite{DBLP:conf/icml/HarshawFWK19}. They show that $\gamma \geq (1 + \frac{s^2}{\sigma^2}\lambda_{\max}(\Lambda))^{-1}$, where $s = \max_{i \in [n]} \|\vx_i\|$. Here we give a bound on the upper weak-submodularity ratio $\beta$.
\begin{theorem}
    Assume a prior distribution $\vtheta \sim \mathcal{N}(0, \Lambda)$, and let $s = \max_{i \in [n]} \|\vx_i\|$. The function $F$ is $\lb 1/c, c \rb$-weakly submodular with $c = 1 + \frac{s^2}{\sigma^2}\lambda_{\max}(\Lambda)$.
    \label{thm:BayesianA}
\end{theorem}
Observe that like for the $R^2$ objective, our upper bound for $\beta$ is the the inverse of the lower bound for $\gamma$.
\begin{proof}[Proof of Theorem \ref{thm:BayesianA}]
    The lower bound on $\gamma$ is shown is \cite{DBLP:conf/icml/HarshawFWK19}. It remains to prove the upper bound on $\beta$.
    Let $B$ be some set of observations and $A \subseteq B$ with $k = |A|$ and for convenience, define $T = B\setminus A$.
    By the Sherman-Morrisson-Woodbury formula (see Lemma~\ref{lem:Woodburry}), we have
    \begin{align}
        F(B) - F(A)
& = \tr( M_A^{-1}) - \tr(M_{B}^{-1}) \notag \\
        & = \tr\bigl( (M_{B} - \sigma^{-2} X_T X_T^T)^{-1}\bigr) - \tr( M_{B}^{-1}) \notag \\
        & = \tr\bigl(M_{B}^{-1} + M_{B}^{-1}X_T(\sigma^2 I - X_T^T M_{B}^{-1} X_T)^{-1}X_T^T M_{B}^{-1}\bigr) - \tr(M_{B}^{-1}) \notag \\
        & = \tr\bigl(M_{B}^{-1}X_T(\sigma^2 I - X_T^T M_{B}^{-1} X_T)^{-1}X_T^T M_{B}^{-1}\bigr) \notag \\
        & = \tr\bigl((\sigma^2 I - X_T^T M_{B}^{-1} X_T)^{-1}X_T^T M_{B}^{-2}X_T\bigr).
        \label{Bayesian_eq:woodbury}
    \end{align}
    The third equality uses the linearity of the trace while the last equality uses the cyclic property of the trace. 
We use the previous equation to derive an upper and lower bound for the numerator and denominator of the submodularity ratio respectively.
    Applying \eqref{Bayesian_eq:woodbury} with $A = B \setminus \{i\}$ (and so $T = \{i\}$) we obtain
    \begin{equation*}
        F(B) - F(B - i) = \frac{\tr(\vx_i^T M_{B}^{-2} \vx_i)}{\sigma^2 - \vx_i^T M_{B}^{-1} \vx_i}.
    \end{equation*}
    Let $\preceq$ be the Loewner ordering of positive semidefinite matrices, where $A \preceq B$ if and only if $B - A \succeq 0$ . First, observe that $\Lambda^{-1} \preceq M_{R}$ for any set $R$, which implies that $\Lambda \succeq M_{R}^{-1}$. Using a second time the Sherman-Morrison-Woodbury formula (Lemma~\ref{lem:Woodburry}) together with the previous observation, we get
    \begin{align*}
        \lb \sigma^2 - \vx_i^T M_{B}^{-1} \vx_i \rb^{-1} & = \sigma^{-2} + \sigma^{-4} \vx_i^T \lb M_{B} - \sigma^{-2} \vx_i \vx_i^T \rb^{-1} \vx_i,\\
        & = \sigma^{-2} + \sigma^{-4} \vx_i^T M_{B\bb\{i\}}^{-1} \vx_i, \\
        & \leq \sigma^{-2} + \sigma^{-4} \vx_i^T \Lambda \vx_i,\\
        & \leq \sigma^{-2} + \sigma^{-4} \lambda_{\max}(\Lambda) s^2,
    \end{align*}
    where $s = \max_{i} \| \vx_i\|_2$ and the last inequality follows by the Courant-Fischer min-max theorem. Summing over all $i \in T = B \setminus A$ and using the linearity of the trace, we have
    \begin{align}
        \sum_{i \in T }F(i | B - i)  = \sum_{i \in T} \frac{\tr(\vx_i^T M_{B}^{-2} \vx_i)}{\sigma^2 - x_i^T M_{B}^{-1} x_i} & \leq \lb \sigma^{-2} + s^2 \sigma^{-4} \lambda_{\max}(\Lambda) \rb \sum_{i \in T} \tr(\vx_i^T M_{B}^{-2} \vx_i) \notag \\
        & = \lb \sigma^{-2} + s^2 \sigma^{-4} \lambda_{\max}(\Lambda) \rb \tr(X_T^T M_{B}^{-2} X_T).
        \label{Bayes_eq:num}
    \end{align}
Returning to the expression of $F(B) - F(A)$, we note that $M_B$ is positive definite, which implies that $M_B^{-1}$ is positive definite. This in turn implies that $-X_T^TM_B^{-1}X_T \preceq 0$ and so $\sigma^2I - X_T^TM_B^{-1}X_T \preceq \sigma^2I$. Thus, $(\sigma^2I - X_T^TM_B^{-1}X_T)^{-1} \succeq \sigma^{-2}I \succ 0$. Therefore,
    \begin{equation*}
        \tr((\sigma^2 I - X_T^T M_{B}^{-1} X_T)^{-1}X_T^T M_{B}^{-2} X_T) \geq \tr{(\sigma^{-2}X_T^T M_{B}^{-2} X_T)} = \sigma^{-2} \tr{(X_T^T M_{B}^{-2} X_T)}.
    \end{equation*}
Combining this with the bound \eqref{Bayes_eq:num}, we have:
    \begin{equation*}
\frac{\sum_{i \in T} F(i | B - i)}{F(B) - F(A)} \leq \frac{\lb \sigma^{-2} + \sigma^{-4} \lambda_{\max}(\Lambda)\cdot s^2 \rb \tr{\lb X_T^T M_{B}^{-2} X_T\rb}}{\sigma^{-2} \tr{(X_T^T M_{B}^{-2} X_T)}} \leq 1 + \frac{s^2}{\sigma^{2}} \lambda_{\max}(\Lambda).
    \end{equation*}
Recalling that $T = B \setminus A$, this completes the proof.
\end{proof}


\section{How large can $\beta$ be?}
\label{sec:how-large-can}

We have shown that the $R^2$ objective (Section \ref{sec:R2}) and the
A-optimal design objective for Bayesian linear regression (Section \ref{sec:optimal-design}) are $(c, 1/c)$-weakly submodular for some parameter $c$. A natural question to ask is whether, given $\gamma > 0$, there is a small non-trivial bound for $\beta$ independent of the size of the ground set. Here we show that this is not true in general, by proving the following claim stated in Section~\ref{sec:non-oblivious-local}:
\betalarge*
  The intuition behind the construction is simple. We build a set function recursively with lower submodularity ratio exactly $\gamma$. The recurrence relation holds until the $(k-1)^\textrm{th}$ marginal, which allows us to have a large value for the final marginal and thus increase $\beta$.

\begin{proof}{Proof of Theorem \ref{lem:worst-case-beta}}
    We start by constructing a monotone set function $f$ on a ground set of $k$ elements. The elements are indistinguishable, meaning that for any given set $S$, two elements $e,e' \in X\bb S$  have the same marginal contribution.
    Therefore, because elements are indistinguishable, the value of a set is a function of its size.
    Let $x_i$ be the value of any set of size $i = 0, 1, \ldots, k$.
    Additionally, let $x_0 = f(\emptyset) = 0$ and $x_k = 1$.
    We define $x_i$ inductively with the following recurrence for $i = 0, 1, \ldots, k-2$:
    \begin{align}
        x_{i+1} = \frac{k -i -\gamma}{k-i} \cdot x_i + \frac{\gamma}{k-i} & \quad \textrm{ or equivalently } \quad x_{i+1} - x_i = \frac{\gamma}{k-i}( 1 - x_i).
        \label{eq_NLP:marginal}
    \end{align}
    It can easily be shown (by induction) that the described sequence is valid, i.e. it is monotone and each $x_i \in [0, 1]$. Additionally, we note that the sequence satisfies:
    \begin{equation}
        1 - x_{i+1} = 1 - \lb \frac{k -i -\gamma}{k-i} \cdot x_i + \frac{\gamma}{k-i} \rb  = \lb 1 - \frac{\gamma}{k-i}\rb \lb 1 - x_i \rb,
        \label{eq_NLP:invariant}
    \end{equation}
for all $i = 0,1,\ldots,k-2$.

First, we show that $f$ has a lower submodularity ratio at most $\gamma$. We prove that for any $B$ and $A \subset B$ such that $|B| = j$ and $|A| = i$:
\begin{equation}
\frac{\sum_{e \in B \setminus A}f(e | A)}{f(B) - f(A)} = \frac{(j-i)(x_{i+1} - x_{i})}{x_j - x_i} \geq \gamma.
\label{eq:gamma-goal}
\end{equation}

First, we consider the case in which  $j = k$. If $i = k-1$, then the left-hand side of \eqref{eq:gamma-goal} is 1. If $i \leq k-2$, then applying the identity \eqref{eq_NLP:marginal}, and recalling that $x_k = 1$ gives:
    \begin{equation*}
        \frac{(k-i)(x_{i+1} - x_i)}{x_k - x_{i}} = \frac{(k-i)\cdot \frac{\gamma}{k-i}(1 - x_i)}{1 - x_i} = \gamma.
    \end{equation*}
for any $i = 1,\ldots,k-1$.

Next, we consider the case in which $j \leq k-1$ and so $i \leq k-2$. Then, by employing recursively the identity \eqref{eq_NLP:invariant} we obtain
    \begin{align}
        x_j - x_i  = (1-x_i) - (1 - x_j) = \lb 1 - x_i \rb \lb 1 - \prod_{\ell = i}^{j - 1}\lb 1 - \frac{\gamma}{k-\ell} \rb \rb.
        \label{eq_NLP:denominator_gamma}
    \end{align}
Since $\gamma \in \ld 0, 1\rd$, we can use the generalized Bernoulli $( 1 - 1/n)^x \leq 1 - x/n$ for $x \in \ld 0, 1 \rd$ to bound each term in the product above. This gives:
    \begin{equation*}
        \prod_{\ell = i}^{j - 1}\lb 1 - \frac{\gamma}{k-\ell} \rb  \geq \prod_{\ell = i}^{j - 1}\lb 1 - \frac{1}{k-\ell} \rb^\gamma
          =  \lb \frac{k -j}{k-i}\rb^\gamma
          = \lb 1  - \frac{j - i}{k-i}\rb^\gamma
          \geq 1 -  \lb \frac{j - i}{k-i} \rb,
    \end{equation*}
where in the last inequality we used the fact that $(1 - x)^\gamma \geq 1 - x$ for $x \in [0, 1]$ and $\gamma \in [0, 1]$. Thus we have:
    \begin{equation*}
        \frac{(j-i)(x_{i+1} - x_i)}{x_j - x_{i}} = \frac{(j-i)\frac{\gamma}{k-i}(1 - x_i)}{(1 - x_i) \lb 1 - \prod_{\ell = i}^{j - 1}\lb 1 - \frac{\gamma}{k-\ell} \rb \rb}
         \geq \frac{(j-i)\frac{\gamma}{k-i}(1-x_i)}{(1-x_i)\frac{j-i}{k-i}}
         = \gamma,
\label{eq:gamma-denom-bound}
    \end{equation*}
where we have used \eqref{eq_NLP:marginal} and \eqref{eq_NLP:denominator_gamma} in the first equation. Combining these cases, we find $f$ is $\gamma$-weakly submodular from below.

To complete the proof we now show that $f$ is not $\beta$-weakly submodular from above for any $\beta < \binom{k-\gamma}{k-1}$. Here we consider the case in which $A = \emptyset$ and $B = X$ and show that:
\begin{equation*}
\frac{\sum_{e \in B}f(e|B-e)}{f(B) - f(\emptyset)} = \frac{k(x_k - x_{k-1})}{x_k - x_{0}} = \binom{k-\gamma}{k-1}.
\end{equation*}
Recall that $x_k = 1$, and $x_0 = 0$, which implies that the denominator is equal to $1$. Recursively applying the identity \eqref{eq_NLP:invariant} gives
    \begin{multline*}
        k(1 - x_{k-1})  = k \cdot \prod_{\ell = 0}^{k-2} \lb 1 - \frac{\gamma}{k -\ell} \rb
        = k \prod_{\ell = 2}^k \lb 1 - \frac{\gamma}{\ell} \rb
        \\ = \frac{\prod_{\ell = 2}^{k} (\ell - \gamma)}{\prod_{\ell= 1}^{k-1} \ell}
        = \frac{\prod_{\ell = 1}^{k-1} (\ell+1-\gamma)}{\prod_{\ell = 1}^{k-1}\ell}
         = \binom{k-\gamma}{k-1},
    \end{multline*}
as required.
\end{proof}




\section{Basic results from linear algebra and probability}
\label{sec:basic-results-from}
We make use of the following basic results related to matrix inverses (see e.g. \cite[Section A.3]{RasmussenWilliamsBook})
\begin{lemma}[Block Matrix Inverse]
Let $B$ and $A$ be matrices of dimension $k \times k$ and $h \times h$, respectively, and let $U$ and $V$ be matrices of dimension $k \times h$ and $h \times k$ respectively. Then,
    \[
    \renewcommand{\arraystretch}{1.4}
    \begin{pmatrix}
    B & U \\
    V & A
    \end{pmatrix}^{-1}
    =
    \begin{pmatrix}
    B^{-1} + B^{-1}USV B^{-1} &  -B^{-1}US \\
    -SVB^{-1} & S
    \end{pmatrix}
    \]
    where $S = (A - V B^{-1} U)^{-1}$ is the Schur complement of $B$.
    \label{thm:block-inverse}
\end{lemma}

\begin{lemma}[Sherman-Morrisson-Woodbury formula]
    \label{lem:Woodburry}
    Let $A, U, C, V$ be matrices of conformable sizes. Then,
    \begin{equation*}
        \lb A + UCV \rb^{-1} = A^{-1} - A^{-1}U\lb C^{-1} + VA^{-1}U \rb^{-1}VA^{-1}.
    \end{equation*}
\end{lemma}

We also make use of the following bound on eigenvalues of normalized covariance matrices shown in~Lemmas 13, 14, and 15 of \cite{DBLP:journals/jmlr/DasK18}:
\begin{lemma}
\label{lem:das-kempe-eigenvalues}
Let $\calL$ and $\calS = \lc X_1, X_2, \ldots, X_n \rc$ be two disjoint sets of zero-mean random variables each of which has variance at most $1$. Let $C$ be the covariance matrix of the set $\calL \cup \calS$. Let $C_\rho$ be the covariance matrix of the set $\lc \res(X_1, \calL), \ldots, \res(X_n, \calL) \rc$ after normalization of the random variables to have unit variance. Then $\lambda_{\min}(C_\rho) \geq \lambda_{\min}(C)$.
\end{lemma}

Finally, we use the following form of the Chernoff bound given in  \cite[Theorem A.1.16]{AlonSpencerBook}:
\begin{lemma}[Chernoff Bound]
\label{lem:chernoff-bound}
Let $X_i$, $1 \leq i \leq n$ be mutually independent random variables with $\expect[X_i] = 0$ and $|X_i| \leq 1$ for all $i$. Set $S = X_1 + \cdots + X_n$. Then for any $a$,
\begin{equation*}
\Pr [S > a] < e^{-a^2/2n}
\end{equation*}
\end{lemma}

\end{document}